\documentclass[11pt]{amsart}
\def\isdraft{0}

\usepackage[dvipsnames]{xcolor}
\usepackage{bbm}
\usepackage{graphicx}
\usepackage{bbm}
\usepackage{amsaddr} 
\usepackage{mathtools}
\usepackage{amsmath,amssymb,amsfonts,dsfont}
\usepackage{prettyref} 
\usepackage[utf8]{inputenc}
\usepackage{enumitem}
\usepackage[hyphens]{url} 
\usepackage{tikz-cd}
\usepackage{tikz}
\usetikzlibrary{positioning}
\usetikzlibrary{arrows}
\usepackage{bussproofs}
\usepackage[mathscr]{euscript}
\usepackage{hyperref}
\usepackage{amsthm}
\usepackage{theapa}
\usepackage{a4wide}
\usepackage[color=black,textcolor=white\if\isdraft0,disable\fi]{todonotes}
\usepackage{stackengine}
\usepackage{stmaryrd}
\usepackage{wrapfig}
\usepackage{pdfpages}
\usepackage{marginnote}
\usepackage{float}
\graphicspath{{fig/}}
\usetikzlibrary{arrows,automata,topaths,matrix,positioning,fit}
\usepgflibrary{shapes.geometric}
\usetikzlibrary{shapes.geometric}
\tikzset{every state/.style={minimum size=0pt}}
\usepackage{footnote}
\usepackage{float}
\usepackage{tabularx}
\usepackage{txfonts}

\newtheorem{theorem}{Theorem}

\newtheorem{corollary}[theorem]{Corollary}
\newtheorem{fact}[theorem]{Fact}
\newtheorem{lemma}[theorem]{Lemma}
\newtheorem{proposition}[theorem]{Proposition}

\theoremstyle{definition} 

\newtheorem{convention}[theorem]{Convention}
\newtheorem{definition}[theorem]{Definition}

\newtheorem{example}[theorem]{Example}

\newtheorem{notation}[theorem]{Notation}

\newtheorem{pseudocode}[theorem]{Pseudocode}

\newtheorem{remark}[theorem]{Remark}

\newrefformat{§}{§\ref{#1}}
\newrefformat{algorithm}{Algorithm \ref{#1}}
\newrefformat{a}{Answer \ref{#1}}
\newrefformat{appendix}{Appendix \ref{#1}}
\newrefformat{claim}{Claim \ref{#1}}
\newrefformat{conclusion}{Conclusion \ref{#1}}
\newrefformat{convention}{Convention \ref{#1}}
\newrefformat{conjecture}{Conjecture \ref{#1}}
\newrefformat{c}{Corollary \ref{#1}}
\newrefformat{equ}{(\ref{#1})}
\newrefformat{e}{Example \ref{#1}}
\newrefformat{exe}{Exercise \ref{#1}}
\newrefformat{d}{Definition \ref{#1}}
\newrefformat{done}{Done \ref{#1}}
\newrefformat{f}{Fact \ref{#1}}
\newrefformat{fig}{Figure \ref{#1}}
\newrefformat{h}{Hypothesis \ref{#1}}
\newrefformat{idea}{Idea \ref{#1}}
\newrefformat{i}{Item \ref{#1}}
\newrefformat{l}{Lemma \ref{#1}}
\newrefformat{note}{Note \ref{#1}}
\newrefformat{n}{Notation \ref{#1}}
\newrefformat{o}{Observation \ref{#1}}
\newrefformat{problem}{Problem \ref{#1}}
\newrefformat{p}{Proposition \ref{#1}}
\newrefformat{pseudocode}{Pseudocode \ref{#1}}
\newrefformat{Q}{Q. \ref{#1}}
\newrefformat{r}{Remark \ref{#1}}
\newrefformat{table}{Table \ref{#1}}
\newrefformat{t}{Theorem \ref{#1}}
\newrefformat{Todo}{Todo \ref{#1}}
\newrefformat{w}{Warning \ref{#1}}

\newcommand{\righttherefore}{:\joinrel\cdot\,}

\title{
    Analogical proportions II
}
\author{
    Christian Anti\'c
}
\address{
    christian.antic@icloud.com\\
    Vienna University of Technology\\
    Vienna, Austria
}

\begin{document}


\begin{abstract} 
    Analogical reasoning is the ability to detect parallels between two seemingly distant objects or situations, a fundamental human capacity used for example in commonsense reasoning, learning, and creativity which is believed by many researchers to be at the core of human and artificial general intelligence. Analogical proportions are expressions of the form ``$a$ is to $b$ what $c$ is to $d$'' at the core of analogical reasoning. The author has recently introduced an abstract algebraic framework of analogical proportions within the general setting of universal algebra. It is the purpose of this paper to further develop the mathematical theory of analogical proportions within that framework as motivated by the fact that it has already been successfully applied to logic program synthesis in artificial intelligence.
\end{abstract}

\maketitle

\section{Introduction}

\textbf{Analogical reasoning} is the ability to detect parallels between two seemingly distant objects or situations, a fundamental human capacity used for example in commonsense reasoning, learning, and creativity which is believed by many researchers to be at the core of human and artificial general intelligence \cite<see e.g.>{Gentner01,Gentner12,Gust08,Hofstadter13,Krieger03,Polya54}. Notable models of analogical reasoning are \citeS{Gentner83} prominent \textit{structure-mapping theory} (SMT) and its implementation within the \textit{structure-mapping engine} \cite{Falklenhainer89} and \citeS{Hofstadter95a} \textit{copycat} algorithm. A formal model using second-order logic similar to SMT is \textit{heuristic-driven theory projection} \cite{Gust06}. \citeA{Winston80} is a classic paper demonstrating the use of analogy in learning. A formalization of analogical reasoning in law \cite{Koszowski19} using Gentzen's sequent calculus is provided by \citeA{Baaz05}. For a short and somewhat outdated introduction to analogical reasoning we refer the reader to \shortciteA{Prade14a}, and for a historic account of models of analogical reasoning we refer the reader to \citeA{Hall89}.

\textbf{Analogical proportions} are expressions of the form ``$a$ is to $b$ what $c$ is to $d$'' --- written $a:b::c:d$ --- at the core of analogical reasoning. Formal models of analogical proportions started to appear only recently, most notably \citeS{Lepage01,Lepage03} axiomatic approach in the linguistic setting, \shortciteS{Miclet09} logical approach in the 2-element boolean setting \shortcite<cf.>{Prade13,Prade18}, and \citeS{Stroppa06} and \shortciteS{Gust06} algebraic approaches. \citeA{Barbot19} use analogical proportions to formalize analogies between concepts. See \citeA{Prade21} for a short summary of applications of analogical proportions to AI \cite<and see>{Correa12}.

The author has recently introduced an abstract algebraic framework of analogical proportions in the general setting of universal algebra \cite{Antic22}. It is \textbf{justification-based} in nature and thus part of the emerging field of Explainable AI \cite<see e.g.>{Heder23}. The \textbf{purpose of this paper} is to contribute to the \textbf{mathematical foundations of analogical reasoning} by further developing the mathematical theory of that framework. This is \textbf{motivated} by the fact that it has already been successfully applied to \textbf{logic program synthesis} for automatic programming and artificial intelligence in \citeA{Antic23-23}, and by the fact that it is capable of capturing two different prominent modellings of a \textbf{boolean proportion} given by \citeA{Klein82} and \citeA{Miclet09} in a single framework \cite{Antic21-3}. For an analysis of the framework in monounary algebras see \citeA{Antic22-2}.
\todo[inline]{Add additional self references in final version if available}

Specifically, we first observe in \prettyref{§:J} that sets of justifications are principal filters, which motivates a peculiar change of notation. We then introduce some terminology which makes the Uniqueness Lemma and Functional Proportion Theorem in \citeA{Antic22} --- which are essential cornerstones of the framework --- easier to apprehend.

We then prove in \prettyref{§:HT} a \textbf{Homomorphism Theorem} as a generalization of the First Isomorphism Theorem in \citeA{Antic22} showing that \textit{arrow} proportions --- which are expressions of the form $a\to b\righttherefore c\to d$ (read as ``$a$ transforms into $b$ as $c$ transforms into $d$'', see \prettyref{§:P}) --- are compatible with homomorphisms.

In \prettyref{§:kl}, we initiate the study of \textbf{fragments} of the framework where the form of justifications is syntactically restricted. Particularly, in \prettyref{§:MAAP}--\ref{§:MMAPs}, we show that we can capture \textbf{difference and geometric proportions} already in the simplest \textit{monolinear} fragment consisting only of justifications with at most one occurrence of a single variable. This implies that \citeS{Stroppa06} notions of arithmetical proportions coincides with monolinear arithmetical proportions in our framework and is thus too restrictive to be considered \textit{the} notion of arithmetical proportions. Moreover, in \prettyref{§:MWP} we study monolinear word proportions.

Analogical proportions between words have found applications to computational linguistics and natural language processing and have been studied in that context by a number of authors \cite<see e.g.>{Hofstadter95a,Lepage01,Lepage03,Lim21,Stroppa06}. In \prettyref{§:WP}, we therefore study \textbf{word proportions} where in \prettyref{t:Stroppa06} and \prettyref{e:Stroppa06} we show that our framework interpreted in the word domain strictly generalizes the often used notion of a word proportion in \citeA{Stroppa06}. This is particularly interesting as our framework has not been geared towards the word domain.
\todo[inline]{Mention Lepage's postulates in the word domain and that we refute some of them justified by simple counterexamples}

Trees are often used in computational linguistics and natural language processing for the representation of the internal structure of complex words and sentences and in first-order predicates used in the representation of the semantics of linguistic entities \cite<see e.g. §2.2 in>{Stroppa06}. In \prettyref{§:TP}, we therefore study \textbf{tree proportions} and show that classical syntactic anti-unification \cite{Reynolds70,Plotkin70,Huet76} can be used to compute tree proportions via a simple syntactic check (\prettyref{t:pqru}). 

Motivated by the previous observation, in \prettyref{§:AU} we show that algebraic \textbf{anti-unification} --- which the author has recently introduced within the setting of universal algebra in \citeA{Antic23-19} --- is related to analogical proportions by giving an illustrative \prettyref{e:20_4_30_x}. We shall point out that other authors have noted the connection between analogical proportions and anti-unification in other settings before: \shortciteA{Krumnack07} is a paper dealing with restricted higher-order anti-unification and analogy making, and \citeA{Weller07} use anti-unification for computing analogical proportions using regular tree grammars.

Finally, in \prettyref{§:FA} we show that analogical proportions can be computed in \textbf{\textit{finite} algebras} using \textbf{tree automata} where we provide algorithms for the most important computational tasks.

\section{Preliminaries}\label{§:P}

This section recalls the abstract algebraic framework of analogical proportions in \citeA{Antic22} where we assume the reader to be fluent in basic universal algebra as it is presented for example in \citeA[§II]{Burris00}.

We denote the \textit{\textbf{empty word}} by $\varepsilon$. As usual, we denote the set of all words over an alphabet $A$ by $ A^\ast$ and define $ A^+:= A^\ast\cup\{\varepsilon\}$. The \textit{\textbf{reverse}} of a word $\mathbf a=a_1\ldots a_n$, $n\geq 1$, is given by $\mathbf a^r:=a_n\ldots a_1$.

A \textit{\textbf{language}} $L$ of algebras is a set of \textit{\textbf{function symbols}}\footnote{We omit constant symbols as we identify constants with 0-ary functions.} together with a \textit{\textbf{rank function}} $r:L\to\mathbb N$, and a denumerable set $X$ of \textit{\textbf{variables}} distinct from $L$. Terms are formed as usual from variables and function symbols.

An \textit{\textbf{$L$-algebra}} $\mathfrak A$ consists of a non-empty set $A$, the \textit{\textbf{universe}} of $\mathfrak A$, and for each function symbol $f\in L$, a function $f^\mathfrak A:A^{r(f)}\to A$, the \textit{\textbf{functions}} of $\mathfrak A$ (the \textit{\textbf{distinguished elements}} of $\mathfrak A$ are the 0-ary functions). Every term $s$ induces a function $s^ \mathfrak A$ on $\mathfrak A$ in the usual way.  We call a term $s$ \textit{\textbf{injective}} in $\mathfrak A$ iff $s^ \mathfrak A$ is an injective function. We denote the variables occurring in $s$ by $X(s)$.

Given a language $L$ and set of variables $X$, the \textit{\textbf{term algebra}} $\mathfrak T_{L,X}$ over $L$ and $X$ has as universe the set $T_{L,X}$ of all terms over $L$ and $X$ and each function symbol $f\in L$ is interpreted by itself, that is,
\begin{align*} 
    f^{\mathfrak T_{L,X}}(p_1,\ldots,p_{r(f)}):=f(p_1,\ldots,p_{r(f)}).
\end{align*} In the sequel, we do not distinguish between a function symbol $f$ and its induced function $f^{\mathfrak T_{L,X}}$ and we call both \textit{\textbf{term functions}}.

We say that a term $t$ is a \textit{\textbf{generalization}} of $s$ --- in symbols, $s\lesssim t$ --- iff there is a substitution $\sigma$ such that $s=t\sigma$. The generalization ordering is reflexive and transitive and thus a pre-order. It is well-known that any two terms $p$ and $q$ have a \textit{\textbf{least general generalization}} \cite{Plotkin70,Reynolds70} in $\mathfrak T_{L,X}$ --- in symbols, $p\Uparrow q$ --- computed by \citeS{Huet76} algorithm as follows. Given an injective mapping $\chi:T_{L,X}\times T_{L,X}\to X$ and two $L$-terms $p,q\in T_{L,X}$, if
\begin{align*} 
    p=f(p_1,\ldots,p_{r(f)}) \quad\text{and}\quad q=f(q_1,\ldots,q_{r(f)}),
\end{align*} for some $f\in L$ and $p_1,\ldots,p_{r(f)},q_1,\ldots,q_{r(f)}\in T_{L,X}$, then define
\begin{align*} 
    p\Uparrow_\chi q:=f(p_1,\ldots,p_{r(f)})\Uparrow_\chi f(q_1,\ldots,q_{r(f)}):=f(p_1\Uparrow_\chi q_1,\ldots,p_{r(f)}\Uparrow_\chi q_{r(f)});
\end{align*} otherwise define
\begin{align*} 
    p\Uparrow_\chi q &:= p\chi q.
\end{align*} Notice that since each constant symbol $a$ induces an 0-ary term function $a(\;)$, we have
\begin{align*} 
    a\Uparrow_\chi a=a(\;)\Uparrow_\chi a(\;)=a.
\end{align*}

\begin{convention} We will always write $s\to t$ instead of $(s,t)$, for any pair of $L$-terms $s$ and $t$ such that every variable in $t$ occurs in $s$, that is, $X(t)\subseteq X(s)$. We call such expressions \textit{\textbf{$L$-rewrite rules}} or \textit{\textbf{$L$-justifications}} where we often omit the reference to $L$. We denote the set of all $L$-justifications with variables among $X$ by $J_{L,X}$. We make the convention that $\to$ binds weaker than every other algebraic operation.
\end{convention}

\begin{definition}\label{d:abcd} We define the \textit{\textbf{analogical proportion relation}} as follows:
\begin{enumerate}
    \item Define the \textit{\textbf{set of justifications}} of an \textit{\textbf{arrow}} $a\to b$ in $\mathfrak A$ by\footnote{For a sequence of objects $\textbf{o}=o_1\ldots o_n$ define $|\textbf{o}|:=n$.}
    \begin{align*} 
        Jus_\mathfrak A(a\to b):=\left\{s\to t\in J_{L,X} \;\middle|\; a\to b=s^ \mathfrak A(\textbf{o})\to t^\mathfrak A\textbf{o},\text{ for some }\textbf{o}\in A^{|\mathbf x|}\right\},
    \end{align*} extended to an \textit{\textbf{arrow proportion}} $a\to b\righttherefore c\to d$\footnote{Read as ``$a$ transforms into $b$ as $c$ transforms into $d$''.} in $\mathfrak{(A,B)}$ by
    \begin{align*} 
        Jus_\mathfrak{(A,B)}(a\to b\righttherefore c\to d):=Jus_\mathfrak A(a\to b)\cap Jus_\mathfrak B(c\to d).
    \end{align*} 

    \item A justification is \textit{\textbf{trivial}} in $\mathfrak{(A,B)}$ iff it justifies every arrow proportion in $\mathfrak{(A,B)}$ and we denote the set of all such trivial justifications by $\emptyset_ \mathfrak{(A,B)}$ simply written $\emptyset$ in cases where the underlying algebras are understood from the context. Moreover, we say that $J$ is a \textit{\textbf{trivial set of justifications}} in $\mathfrak{(A,B)}$ iff every justification in $J$ is trivial. 

    \item Now we say that $a\to b\righttherefore c\to d$ \textit{\textbf{holds}} in $\mathfrak{(A,B)}$ --- in symbols,
    \begin{align*} 
        a\to b\righttherefore_\mathfrak{(A,B)}\, c\to d
    \end{align*} iff
    \begin{enumerate}
        \item either $Jus_\mathfrak A(a\to b)\cup Jus_\mathfrak B(c\to d)=\emptyset_ \mathfrak{(A,B)}$ consists only of trivial justifications, in which case there is neither a non-trivial relation from $a$ to $b$ in $\mathfrak A$ nor from $c$ to $d$ in $\mathfrak B$; or
        \item $Jus_\mathfrak{(A,B)}(a\to b\righttherefore c\to d)$ is maximal with respect to subset inclusion among the sets $Jus_\mathfrak{(A,B)}(a\to b\righttherefore c\to d')$, $d'\in B$, containing at least one non-trivial justification, that is, for any element $d'\in \mathfrak B$,\footnote{We ignore trivial justifications and write ``$\emptyset\subsetneq\ldots$'' instead of ``$\{\text{trivial justifications}\}\subsetneq\ldots$'' et cetera.}
        \begin{align*} 
            \emptyset_ \mathfrak{(A,B)}\subsetneq Jus_\mathfrak{(A,B)}(a\to b\righttherefore c\to d)&\subseteq Jus_\mathfrak{(A,B)}(a\to b\righttherefore c\to d')
        \end{align*} implies
        \begin{align*} 
            \emptyset_ \mathfrak{(A,B)}\subsetneq Jus_\mathfrak{(A,B)}(a\to b\righttherefore c\to d')\subseteq Jus_\mathfrak{(A,B)}(a\to b\righttherefore c\to d).
        \end{align*} We abbreviate the above definition by simply saying that $Jus_\mathfrak{(A,B)}(a\to b\righttherefore c\to d)$ is \textit{\textbf{$d$-maximal}}.
    \end{enumerate}

    \item Finally, the analogical proportion entailment relation is most succinctly defined by
    \begin{align*} 
        a:b::_ \mathfrak{(A,B)}c:d \quad:\Leftrightarrow\quad 
            &a\to b\righttherefore_ \mathfrak{(A,B)}\, c\to d \quad\text{and}\quad b\to a\righttherefore_ \mathfrak{(A,B)}\, d\to c\\
            &c\to d\righttherefore_ \mathfrak{(B,A)}\, a\to b \quad\text{and}\quad d\to c\righttherefore_ \mathfrak{(B,A)}\, b\to a.
    \end{align*}
\end{enumerate} We will always write $\mathfrak A$ instead of $\mathfrak{AA}$. Moreover, we define
\begin{align*} 
    \mathscr S_ \mathfrak{(A,B)}(a:b::c: \mathfrak x):=\{d\in B\mid a:b::_ \mathfrak{(A,B)}c:d\}.
\end{align*}
\end{definition}

\begin{example}[\citeA{Antic22}, Example 11] First consider the algebra $\mathfrak A_1:=(\{a,b,c,d\})$, consisting of four distinct elements with no functions and no constants:
\begin{center}
\begin{tikzpicture} 
    \node (a)               {$a$};
    \node (b) [above=of a]  {$b$};
    \node (c) [right=of a]  {$c$};
    \node (d) [above=of c]  {$d$};
\end{tikzpicture}
\end{center} Since $Jus_{\mathfrak A_1}(a'\to b')\cup Jus_{\mathfrak A_1}(c'\to d')$ contains only trivial justifications for \textit{any distinct} elements $a',b',c',d'\in A'$, we have, for example:
\begin{align*} 
    a:b::_{\mathfrak A_1}\,c:d \quad\text{and}\quad a:c::_{\mathfrak A_1}\,b:d.
\end{align*} On the other hand, since
\begin{align*} 
    Jus_{\mathfrak A_1}(a\to a)\cup Jus_{\mathfrak A_1}(a\to d)=\{x\to x\}\neq\emptyset
\end{align*} and
\begin{align*} 
    \emptyset=Jus_{\mathfrak A_1}(a\to a\righttherefore a\to d)\subsetneq Jus_{\mathfrak A_1}(a\to a\righttherefore a\to a)=\{x\to x\},
\end{align*} we have
\begin{align*} 
    a\to a\righttherefore_{\mathfrak A_1}\, a\to d,
\end{align*} which implies
\begin{align*} 
    a:a::_{\mathfrak A_1}a:d.
\end{align*}

Now consider the slightly different algebra $\mathfrak A_2:=(\{a,b,c,d\},f)$, where $f$ is the unary function defined by 
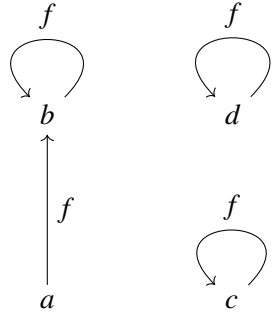
\begin{center}
\begin{tikzpicture} 
\node (a)               {$a$};
\node (b) [above=of a,yshift=1cm]  {$b$};
\node (c) [right=of a,xshift=1cm]  {$c$};
\node (d) [right=of b,xshift=1cm]  {$d$};

\draw[->] (a) to [edge label'={$f$}] (b);
\draw[->] (b) to [edge label'={$f$}] [loop] (b);
\draw[->] (c) to [edge label'={$f$}] [loop] (c);
\draw[->] (d) to [edge label'={$f$}] [loop] (d);
\end{tikzpicture}
\end{center} We expect the proportion $a:b::c:d$ to fail in $\mathfrak A_2$ as it has no non-trivial justification. In fact, 
\begin{align*} 
    Jus_{\mathfrak A_2}(a\to b)\cup Jus_{\mathfrak A_2}(c\to d)=\left\{x\to f^nx \;\middle|\; n\geq 1\right\}\neq\emptyset 
\end{align*} and
\begin{align*} 
    Jus_{\mathfrak A_2}(a\to b\righttherefore c\to d)=\emptyset
\end{align*} show
\begin{align*} 
    a:b::_{\mathfrak A_2}c:d.
\end{align*}

In the algebra $\mathfrak A_3$ given by 
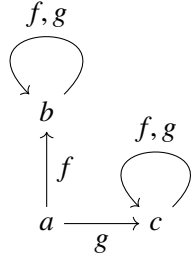
\begin{center}
\begin{tikzpicture} 
    \node (a)               {$a$};
    \node (b) [above=of a]  {$b$};
    \node (c) [right=of a]  {$c$};
    \draw[->] (a) to [edge label'={$f$}] (b);
    \draw[->] (a) to [edge label'={$g$}] (c);
    \draw[->] (b) to [edge label'={$f,g$}] [loop] (b);
    \draw[->] (c) to [edge label'={$f,g$}] [loop] (c);
\end{tikzpicture}
\end{center} we have
\begin{align*} 
    a:b\not::_{\mathfrak A_3}a:c.
\end{align*} The intuitive reason is that $a:b::a:b$ is a more plausible proportion than $a:b::a:c$, which is reflected in the computation
\begin{align*} 
    \emptyset=Jus_{\mathfrak A_3}(a\to b\righttherefore a\to c)\subsetneq Jus_{\mathfrak A_3}(a\to b\righttherefore a\to b)=\{x\to f(x),\ldots\}.
\end{align*}
\end{example}

Computing all justifications of an arrow proportion is difficult in general, which fortunately can be omitted in many cases: 

\begin{definition} We call a set $J$ of justifications a \textit{\textbf{characteristic set of justifications}} of $a\to b\righttherefore c\to d$ in $\mathfrak{(A,B)}$ iff $J$ is a sufficient set of justifications in the sense that
\begin{enumerate}
    \item $J\subseteq Jus_\mathfrak{(A,B)}(a\to b\righttherefore c\to d)$, and
    \item $J\subseteq Jus_\mathfrak{(A,B)}(a\to b\righttherefore c\to d')$ implies $d'=d$, for each $d'\in B$.
\end{enumerate} In case $J=\{s\to t\}$ is a singleton set satisfying both conditions, we call $s\to t$ a \textit{\textbf{characteristic justification}} of $a\to b\righttherefore c\to d$ in $\mathfrak{(A,B)}$.
\end{definition}

In the tradition of the ancient Greeks, \citeA{Lepage03} introduced (in the linguistic context) a set of proportional properties as a guideline for formal models of analogical proportions. His proposed list has since been extended by a number of authors \cite<e.g.>{Miclet09,Prade13,Barbot19,Antic22} 
and can now be summarized as follows:\footnote{\citeA{Lepage03} uses different names for the axioms --- we have decided to remain consistent with the nomenclature in \citeA[§4.2]{Antic22}.}
\begin{align*}
    &a:b::_ \mathfrak A a:b \quad\text{(p-reflexivity)},\\
    &a:b::_\mathfrak{(A,B)} c:d \quad\Leftrightarrow\quad c:d::_{(\mathfrak{B,A})}a:b\quad\text{(p-symmetry)},\\
    &a:b::_\mathfrak{(A,B)} c:d \quad\Leftrightarrow\quad b:a::_\mathfrak{(A,B)}d:c\quad\text{(inner p-symmetry)},\\
    &a:a::_\mathfrak A a:d \quad\Leftrightarrow\quad d=a\quad\text{(p-determinism)},\\
    &a:a::_\mathfrak{(A,B)} c:c \quad\text{(inner p-reflexivity)},\\
    &a:b::_ \mathfrak A c:d \quad\Leftrightarrow\quad a:c::_ \mathfrak A b:d \quad\text{(central permutation)},\\
    &a:a::_ \mathfrak A c:d \quad\Rightarrow\quad d=c \quad\text{(strong inner p-reflexivity)},\\
    &a:b::_ \mathfrak A a:d \quad\Rightarrow\quad d=b \quad\text{(strong p-reflexivity)}.
\end{align*} Moreover, the following property is considered, for $a,b\in A\cap B$:
\begin{align*}
    a:b::_\mathfrak{(A,B)} b:a\quad\text{(p-commutativity).}
\end{align*} 

Furthermore, the following properties are considered, for $L$-algebras $\mathfrak{A,B,C}$ and elements $a,b\in A$, $c,d\in B$, $e,f\in C$:
\begin{prooftree}
    \AxiomC{$a:b::_\mathfrak{(A,B)} c:d$}
        \AxiomC{$c:d::_\mathfrak{(B,C)} e:f$}
        \RightLabel{(p-transitivity),}
    \BinaryInfC{$a:b::_\mathfrak{(A,C)} e:f$}
\end{prooftree} and, for elements $a,b,e\in A$ and $c,d,f\in B$, the property
\begin{prooftree}
    \AxiomC{$a:b::_\mathfrak{(A,B)} c:d$}
    \AxiomC{$b:e::_\mathfrak{(A,B)} d:f$}
    \RightLabel{(inner p-transitivity),}
    \BinaryInfC{$a:e::_\mathfrak{(A,B)} c:f$}
\end{prooftree} and, for elements $a\in A$, $b\in A\cap B$, $c\in B\cap C$, and $d\in C$, the property
\begin{prooftree}
    \AxiomC{$a:b::_\mathfrak{(A,B)}b:c$}
    \AxiomC{$b:c::_\mathfrak{(B,C)}c:d$}
    \RightLabel{(central p-transitivity).}
    \BinaryInfC{$a:b::_\mathfrak{(A,C)}c:d$}
\end{prooftree} Notice that central p-transitivity follows from p-transitivity. 

Finally, the following schema is considered, where $\mathfrak A'$ and $\mathfrak B'$ are $L'$-algebras, for some language $L\subseteq L'$:
\begin{prooftree}
    \AxiomC{$a:b::_\mathfrak{(A,B)}c:d$}
        \AxiomC{$\mathfrak A=\mathfrak A'\upharpoonright L$}
            \AxiomC{$\mathfrak B=\mathfrak B'\upharpoonright L$}
            \RightLabel{(p-monotonicity).}
    \TrinaryInfC{$a:b::_\mathfrak{(A',B')}c:d$}
\end{prooftree}

\section{Justifications}\label{§:J}

In this section, we observe that sets of justifications are principal filters, which will motivate a change of notation by replacing $Jus$ by $\uparrow$ thus expressing syntactically the close connection to generalizations more adequately.

Recall that a \textit{\textbf{filter}} $F$ on a pre-ordered set $(P,\leq)$ is a subset of $P$ satisfying:
\begin{enumerate}
    \item $F$ is non-empty.
    \item $F$ is downward directed, that is, for every $a,b\in F$, there is some $c\in F$ such that $c\leq a,b$.
    \item $F$ is an upper set or upward closed, that is, for every $a\in F$ and $b\in P$, if $a\leq b$ then $b\in F$.
\end{enumerate} The smallest filter containing an element $a$ is a \textit{\textbf{principal filter}} and $a$ is a \textit{\textbf{principal element}} --- it is given by
\begin{align*} 
    \uparrow_{(P,\leq)}a:=\{b\in P\mid a\leq b\}.
\end{align*}

We extend the generalization pre-ordering from terms to justification via
\begin{align*} 
   s\to t\;\lesssim\; s'\to t' \quad\Leftrightarrow\quad s\lesssim s' \quad\text{and}\quad t\lesssim t'.
\end{align*}

\begin{fact}\label{f:principal} The set of all generalizations of a term forms a principal filter with respect to the generalization pre-ordering generated by that term. Moreover, the set of all justifications of an arrow forms a principal filter with respect to the generalization pre-ordering generated by that justification.
\end{fact}

\begin{notation} \prettyref{f:principal} motivates the following notation which we will use in the rest of the paper:
\begin{align*} 
    \uparrow_ \mathfrak A(a\to b):=Jus_ \mathfrak A(a\to b),
\end{align*} extended to an arrow proportion by
\begin{align*} 
    \uparrow_\mathfrak{(A,B)}(a\to b \righttherefore c\to d):=Jus_\mathfrak{(A,B)}(a\to b \righttherefore c\to d).
\end{align*}
\end{notation}

We shall now reformulate some key results in \citeA{Antic22} using a different --- hopefully more intuitive --- terminology. For this, we first define, for a term $s\in T_{L,X}$ and element $a\in A$, the set
\begin{align*} 
	\langle s,a\rangle_\mathfrak A:=\left\{\textbf{o}\in A^{r(s)} \;\middle|\; a=s^ \mathfrak A(\textbf{o})\right\},
\end{align*} consisting of all solutions to the polynomial equation $a=s( \mathbf x)$ in $\mathfrak A$. We can now depict every justification $s\to t$ of $a\to b\righttherefore c\to d$ as follows \cite<see>[Convention 15]{Antic22}:
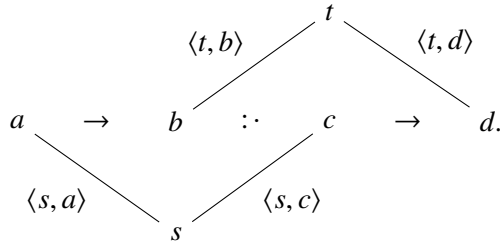
\begin{center}
\begin{tikzpicture}[node distance=1cm and 0.5cm]
\node (a)               {$a$};
\node (d1) [right=of a] {$\to$};
\node (b) [right=of d1] {$b$};
\node (d2) [right=of b] {$\righttherefore $};
\node (c) [right=of d2] {$c$};
\node (d3) [right=of c] {$\to$};
\node (d) [right=of d3] {$d$.};
\node (s) [below=of b] {$s$};
\node (t) [above=of c] {$t$};

\draw (a) to [edge label'={$\langle s,a\rangle$}] (s); 
\draw (c) to [edge label={$\langle s,c\rangle$}] (s);
\draw (b) to [edge label={$\langle t,b\rangle$}] (t);
\draw (d) to [edge label'={$\langle t,d\rangle$}] (t);
\end{tikzpicture}
\end{center} Moreover, we have
\begin{align}\label{equ: 230824-langle} 
	s\to t\in\ \uparrow( a\to b\righttherefore c\to d) \quad\Leftrightarrow\quad \langle s,a\rangle\cap\langle t,b\rangle\neq\emptyset \quad\text{und}\quad \langle s,c\rangle\cap\langle t,d\rangle\neq\emptyset.
\end{align}

Define
\begin{align*} 
	\mathbbm 1_\mathfrak A(s):=\{a\in A\mid |\langle s,a\rangle_\mathfrak A|=1\}.
\end{align*}

We can now reformulate the rather opaque Uniqueness Lemma and Functional Proportion Theorem in \citeA{Antic22} more cleanly using the above notions:

\begin{lemma}[Uniqueness Lemma]\label{l:UL} We have the following implications:
\begin{prooftree}
	\AxiomC{$s\to t\in\ \uparrow_\mathfrak{(A,B)}(a\to b\righttherefore c\to d)$}
		\AxiomC{$c\in\mathbbm 1_\mathfrak B(s)$}
	\BinaryInfC{$a\to b\righttherefore_\mathfrak{(A,B)}\, c\to d$}
\end{prooftree} and
\begin{prooftree}
	\AxiomC{$s\to t\in\ \uparrow_\mathfrak{(A,B)}(a\to b\righttherefore c\to d)$}
		\AxiomC{$a\in\mathbbm 1_\mathfrak A(s)\qquad b\in \mathbbm 1_\mathfrak A(t)\qquad c\in \mathbbm 1_\mathfrak B(s)\qquad d\in \mathbbm 1_\mathfrak B(t)$}
        \RightLabel{.}
	\BinaryInfC{$a:b::_\mathfrak{(A,B)}c:d$}
\end{prooftree}
\end{lemma}

\begin{theorem}[Functional Proportion Theorem]\label{t:FPT} For any $L$-term $t(x)$, we have the following implication:
\begin{prooftree}
	\AxiomC{$a\in\mathbbm 1_\mathfrak A(t)$}
		\AxiomC{$c\in\mathbbm 1_\mathfrak B(t)$}
        \RightLabel{.}
	\BinaryInfC{$a:t^\mathfrak A(a)::_\mathfrak{(A,B)}c:t^\mathfrak B(c)$}
\end{prooftree} In this case, we call $t^\mathfrak B(c)$ a \textit{\textbf{functional solution}} of $a:b::c: \mathfrak x$ in $\mathfrak{(A,B)}$ characteristically justified by $x\to t(x)$. This holds in particular in case $t$ is injective in $\mathfrak A$ and $\mathfrak B$.
\end{theorem}


\section{Homomorphism Theorem}\label{§:HT}

Analogical proportions are compatible with structure-preserving mappings in the following way: The First Isomorphism Theorem in \citeA{Antic22} says that for any isomorphism $H: \mathfrak{A\to B}$,
\begin{align*} 
    a:b::_\mathfrak{(A,B)} Ha:Hb,\quad\text{for all $a,b\in A$}.
\end{align*} A simple counterexample shows that homomorphisms are in general \textit{not} compatible with analogical proportions in the same way (see \prettyref{e:H}). In this section, we shall recover a part of the result by showing that homomorphisms are compatible with \textit{arrow} proportions. We first show an auxiliary lemma (analogous to the Isomorphism Lemma in \citeA{Antic22}):

\begin{lemma}[Homomorphism Lemma]\label{l:HL} For any homomorphism $H: \mathfrak{A\to B}$ and $a,b\in A$,
\begin{align}\label{equ:a->b_subseteq_Ha->Hb}  
    \uparrow_ \mathfrak A(a\to b)\subseteq\ \uparrow_ \mathfrak B(Ha\to Hb).
\end{align} In case $H$ is an isomorphism, we have
\begin{align}
    \uparrow_ \mathfrak A(a\to b)=\ \uparrow_ \mathfrak B(Ha\to Hb).
\end{align}
\end{lemma}
\begin{proof} We have
\begin{align*} 
   s\to t\in\ \uparrow_ \mathfrak A(a\to b) 
        \quad&\Leftrightarrow\quad a\to b=s^ \mathfrak A(\textbf{o})\to t^ \mathfrak A(\textbf{o}),\quad\text{for some $\textbf{o}\in A^{r(s)}$}\\
        \quad&\Rightarrow\quad Ha\to Hb=H(s^ \mathfrak A(\textbf{o}))\to H(t^ \mathfrak A(\textbf{o}))=s^ \mathfrak B(H \textbf{o})\to t^ \mathfrak B(H \textbf{o})\\
        \quad&\Rightarrow\quad s\to t\in\ \uparrow_ \mathfrak B(Ha\to Hb).
\end{align*}

The second part has already been shown in the proof of \citeA[Isomorphism Lemma]{Antic22}.
\end{proof}

\begin{theorem}[Homomorphism Theorem]\label{t:HT} For any homomorphism $H: \mathfrak{A\to B}$ and elements $a,b\in A$, we have the following implication:
\begin{prooftree}
    \AxiomC{$\uparrow_ \mathfrak A(a\to b)=\emptyset_ \mathfrak A \quad\Rightarrow\quad \uparrow_ \mathfrak B(Ha\to Hb)=\emptyset_ \mathfrak B$}
    \RightLabel{.}
    \UnaryInfC{$a\to b \righttherefore_\mathfrak{(A,B)}\, Ha\to Hb$}
\end{prooftree}
\end{theorem}
\begin{proof} By the Homomorphism \prettyref{l:HL}, we have
\begin{align*} 
    \uparrow_{ \mathfrak{(A,B)}}(a\to b\righttherefore Ha\to Hb)=\ \uparrow_ \mathfrak A(a\to b)\ \cap \uparrow_ \mathfrak B(Ha\to Hb)=\ \uparrow_ \mathfrak A(a\to b),
\end{align*} which shows the $Hb$-maximality of $\uparrow_{ \mathfrak{(A,B)}}(a\to b\righttherefore Ha\to Hb)$.

It remains to show that we cannot have
\begin{align*} 
    \uparrow_ \mathfrak A(a\to b)\ \cup\uparrow_ \mathfrak B(Ha\to Hb)\neq\emptyset_ \mathfrak{(A,B)} \quad\text{whereas}\quad \uparrow_{ \mathfrak{(A,B)}}(a\to b\righttherefore Ha\to Hb)=\emptyset_ \mathfrak{(A,B)}.
\end{align*} This is a direct consequence of \prettyref{equ:a->b_subseteq_Ha->Hb} and the assumption that $\uparrow_ \mathfrak A(a\to b)=\emptyset_ \mathfrak A$ implies $\uparrow_ \mathfrak B(Ha\to Hb)=\emptyset_ \mathfrak B$.
\end{proof}

\begin{example}\label{e:H} Let us now analyze the counterexample in \citeA[Example 39]{Antic22}. Let $\mathfrak A:=(\{a,b,c,d\},g)$, $\mathfrak B:=(\{e,f\},g)$, and $H: \mathfrak{A\to B}$ be given by
\begin{center}
\begin{tikzpicture} 
    \node (d) {$d$};
    \node (c) [below=of d,yshift=-1cm]  {$c$};
    \node (b) [below=of c,yshift=-1cm]  {$b$};
    \node (a) [below=of b,yshift=-1cm]  {$a$};
        \node (f) [right=of c,xshift=1cm]  {$f$};
        \node (e) [right=of b,xshift=1cm]  {$e$};
    \draw[->] (a) to [edge label={$g$}] (b);
    \draw[->] (c) to [edge label={$g$}] (d);
    \draw[->] (e) to [edge label'={$g$}] (f);
    \draw[->, dashed] (d) to [edge label={$H$}] (f);
    \draw[->, dashed] (c) to (e);
    \draw[->, dashed] (b) to (f);
    \draw[->, dashed] (a) to [edge label'={$H$}] (e);
    \draw[->] (b) to [edge label'={$g$}] [loop] (b);
    \draw[->] (d) to [edge label'={$g$}] [loop] (d);
        \draw[->] (f) to [edge label'={$g$}] [loop] (f);
\end{tikzpicture}
\end{center} In \citeA[Example 39]{Antic22} it is shown that
\begin{align*} 
    a\to d\not \righttherefore Ha\to Hd.
\end{align*} This is not a contradiction to the Homomorphism \prettyref{t:HT}, since
\begin{align*} 
    \uparrow(a\to d)=\emptyset \quad\not\Rightarrow\quad \uparrow(Ha\to Hd)=\emptyset,
\end{align*} shows that we \textit{cannot} apply the theorem. What we do have is
\begin{align*} 
    a\to b \righttherefore Ha\to Hb \quad\text{and}\quad c\to d \righttherefore Hc\to Hd.
\end{align*} In fact, we even have
\begin{align*} 
    a:b::Ha:Hb \quad\text{and}\quad c:d::Hc:Hd.
\end{align*}
\end{example}

\begin{definition}[\citeA{Antic22-4}] A \textit{\textbf{proportional homomorphism}} is a mapping $H: \mathfrak A\to \mathfrak B$ satisfying
\begin{align*} 
    a:b::_ \mathfrak A c:d \quad\Leftrightarrow\quad Ha:Hb::_ \mathfrak B Hc:Hd,\quad\text{for all $a,b,c,d\in A$}.
\end{align*}
\end{definition}

\begin{theorem}[Isomorphism Theorem] Every isomorphism is a proportional homomorphism.
\end{theorem}
\begin{proof} A direct consequence of the Homomorphism \prettyref{l:HL}.
\end{proof}

\begin{remark} The Isomorphism Theorem shows that isomorphisms are \textit{\textbf{analogy-preserving functions}} in the sense of \citeA[Definition 6]{Couceiro23} which are a variant of the \textit{\textbf{analogical jump}} in \citeA{Davies87}.
\end{remark}

\section{\texorpdfstring{The $(k,\ell)$-fragments}{Fragments}}\label{§:kl}

Since computing the set of \textit{all} justifications is rather difficult in general, it is reasonable to study fragments of the framework. For this, we introduce in this section the $(k,\ell)$-fragments: 

\begin{definition} Let $X_k:=\{x_1,\ldots,x_k\}$, for some $k\in\mathbb N\cup\{\infty\}$, so that $X_\infty=X$. Let $\ell\in \mathbb N\cup\{\infty\}$. Define
\begin{align*} 
    \uparrow^{(k,\ell)}_\mathfrak A a:=(\uparrow_\mathfrak A a)\cap\{s(x_1,\ldots,x_k)\in T_{L,X_k}\mid\text{each of the $k$ variables in $X_k$ occurs at most $\ell$ times in $s$}\}.
\end{align*} We write $k$ instead of $(k,\infty)$ so that
\begin{align*} 
    \uparrow^k_ \mathfrak A a=(\uparrow_ \mathfrak A a)\cap T_{L,X_k}.
\end{align*} We extend the above notions from elements to arrows by
\begin{align*} 
    \uparrow^{(k,\ell)}_\mathfrak A(a\to b):= \left\{s\to t\in\ \uparrow_ \mathfrak A(a\to b) \;\middle|\; s\in\ \uparrow^{(k,\ell)}_ \mathfrak A a\text{ and }t\in\ \uparrow^{(k,\ell)}_ \mathfrak A b \right\},
\end{align*} extended to arrow proportions by
\begin{align*} 
    \uparrow^{(k,\ell)}_\mathfrak{(A,B)}(a\to b \righttherefore c\to d):=\ \uparrow^{(k,\ell)}_\mathfrak A(a\to b)\ \cap \uparrow^{(k,\ell)}_\mathfrak B(c\to d).
\end{align*} The \textit{\textbf{analogical proportion $(k,\ell)$-relation}} $::_{(k,\ell)}$ is defined in the same way as $::$ with $\uparrow$ replaced by $\uparrow^{(k,\ell)}$. In case the underlying algebras are clear from the context, we will often write $a:b::_{(k,\ell)} c:d$ to denote the analogical proportion relation in the $(k,\ell)$-fragment.  The \textit{\textbf{monolinear fragment}} $(1,1)$ consists only of justifications with at most one occurrence of a single variable on each side and is denoted by $m$. 
\end{definition}

\begin{remark} The unrestricted framework occurs as the special ``fragment'' $(\infty,\infty)$ where every justification may have arbitrary many variables occurring arbitrary often.
\end{remark}

\subsection{Monolinear additive arithmetical proportions}\label{§:MAAP}

In this subsection, we work in $(\mathbb{Z,+,Z})$ (see the forthcoming \prettyref{r:(Z,+)}). We begin by noting that the set of monolinear justifications of $a\to b$ is given by
\begin{align*} 
    \uparrow^m(a\to b)=&\left\{k+x\to \ell+x \;\middle|\; a\to b=k+o\to \ell+o,\text{ for some $o\in\mathbb Z$}\right\}\\
        &\cup\{k+x\to b \mid a\to b=k+o\to b,\text{ for some $o\in\mathbb Z$}\}\cup\{a\to b\}.
\end{align*}

\begin{remark}\label{r:(Z,+)} In $(\mathbb Z,+)$ containing no constants, the only monolinear rewrite rule is $x\to x$ which justifies only inner reflexive proportions of the form $a:a::_m c:c$. This explains why we instead consider the algebra $(\mathbb{Z,+,Z})$ in which \textit{every} integer is a distinguished element.
\end{remark}

Interestingly, it turns out that monolinear additive number proportions are characterized by difference proportions:

\begin{theorem}[Difference Proportion Theorem]\label{t:DPT}
\begin{align*} 
    a:b::_{(\mathbb Z,+, \mathbb Z),m} c:d 
        \quad&\Leftrightarrow\quad a=k+o,\quad b=\ell+o,\quad c=k+u,\quad d=\ell+u,\quad k,\ell,o,u\in\mathbb Z\\
        \quad&\Leftrightarrow\quad a-b=c-d\quad\text{\textit{\textbf{(difference proportion)}}}.
\end{align*}
\end{theorem}
\begin{proof} We first show
\begin{align}\label{equ:a=k+o_iff_a-b=c-d} 
    a=k+o,\quad b=\ell+o,\quad c=k+u,\quad d=\ell+u \quad\Leftrightarrow\quad a-b=c-d,
\end{align} for some integers $k,\ell,o,u\in\mathbb Z$. The direction from left to right holds trivially. For the other direction, we proceed as follows. We can always write $a=k+o$ and $b=\ell+o$, for some $k,\ell,o\in\mathbb Z$. We then have $a-b=k-\ell$. Analogously, we can always write $c=k+u$ and $d=\ell'+u$, for some $\ell',u\in\mathbb Z$. We then have $c-d=k-\ell'$. By assumption, we have $a-b=c-d$ which implies $k-\ell=k-\ell'$ and therefore $\ell=\ell'$ and finally $d=\ell+u$.

We now proceed to show the first equivalence in the statement of the theorem:

($\Rightarrow$) By assumption, we have $a\to b\righttherefore_m\, c\to d$ which holds iff either\footnote{Recall from \prettyref{d:abcd} that $\emptyset$ here really means the set of all trivial monolinear justifications $\emptyset_{(\mathbb Z,+, \mathbb Z)}$.}
\begin{align*} 
    \uparrow^m( a\to b)\ \cup \uparrow^m(c\to d)=\emptyset
\end{align*} or $\uparrow^m(a\to b\righttherefore c\to d)$ contains at least one non-trivial monolinear justification and is subset maximal with respect to $d$. In the first case, notice that neither $\uparrow^m(a\to b)$ nor $\uparrow^m(c\to d)$ can consist only of trivial justifications as we always have
\begin{align*} 
    a\to b\in\ \uparrow^m(a\to b) \quad\text{and}\quad c\to d\in\ \uparrow^m(c\to d).
\end{align*} In the second case, by assumption we must have some monolinear justification $s(x)\to t(x)$ of $a\to b\righttherefore c\to d$ in $(\mathbb{Z,+,Z})$. We distinguish the following cases:
\begin{enumerate}
    \item If $s(x)\to t(x)$ equals $a\to b$ or $c\to d$, we must have $a=c$ and $b=d$.
    \item Else if $s(x)\to t(x)$ equals $k+x\to\ell+x$, we must have $a=k+o$, $b=\ell+o$, $c=k+u$, and $d=\ell+u$, for some integers $o,u\in\mathbb Z$, which is equivalent to $a-b=c-d$ by \prettyref{equ:a=k+o_iff_a-b=c-d}.
    \item Else if $s(x)\to t(x)$ equals $k+x\to b$, we must have $b=d$. Then, by assumption, we must also have
    \begin{align*} 
        a\to b \righttherefore_m\, c\to b \quad\text{and}\quad b\to a \righttherefore_m\, b\to c.
    \end{align*} So, either we have
    \begin{align*} 
        \uparrow^m(b\to a)\ \cup \uparrow^m(b\to c)=\emptyset
    \end{align*} or $\uparrow^m(b\to a \righttherefore b\to c)$ is non-empty and subset maximal with respect to $c$. Again, the sets $\uparrow^m(b\to a)$ and $\uparrow^m(b\to c)$ cannot be empty as they certainly contain $b\to a$ and $b\to c$, respectively. Hence, $\uparrow^m(b\to a \righttherefore b\to c)$ must contain at least one non-trivial monolinear justification $s'(x)\to t'(x)$. We distinguish the following cases:
    \begin{enumerate}
        \item If $s'(x)\to t'(x)$ equals $b\to a$ or $b\to c$, we must have $c=a$.
        \item Else if $s'(x)\to t'(x)$ equals $k'+x\to\ell'+x$, we must have $b=k'+o$, $a=\ell'+o$, $b=k'+u$, and $c=\ell'+u$, for some $o,u\in\mathbb Z$, which implies $b=k'+o=k'+u$ and therefore $o=u$ and hence $a=\ell'+o=\ell'+u=c$.
        \item Finally, if $s'(x)\to t'(x)$ equals $k'+x\to a$, we must have $a=c$.
    \end{enumerate}
\end{enumerate}

($\Leftarrow$) Every justification of the form $k+x\to\ell+x$ is a characteristic justification by the Uniqueness \prettyref{l:UL} since $k+x$ and $\ell+x$ are injective in $(\mathbb{Z,+,Z})$. Since $a-b=c-d$ holds by assumption, $x\to x+b-a$ is a characteristic justification of $ a:b::c:d$ in $(\mathbb{Z,+,Z})$.
\end{proof}

Interestingly, additive monolinear number proportions are equivalent to number proportions in the domain of natural numbers $(\mathbb N,S)$ with the \textit{\textbf{successor function}} $S(x):=x+1$.

\begin{corollary}\label{c:(N,S)} $a:b::_{(\mathbb N,+, \mathbb N),m} c:d \quad\Leftrightarrow\quad a:b::_{(\mathbb N,S)} c:d$.
\todo[inline]{check ob man $\mathbb Z$ in \prettyref{t:DPT} durch $\mathbb N$ ersetzen darf}
\end{corollary}
\begin{proof} A direct consequence of the Difference Proportion \prettyref{t:DPT} and the Difference Proportion Theorem in \citeA{Antic22-2}.
\end{proof}

\begin{theorem} All the properties listed in \prettyref{§:P} hold in $(\mathbb{Z,+,Z})$ in the monolinear fragment except for p-commutativity.
\end{theorem}
\begin{proof} We have the following proofs:
\begin{itemize}
    \item The proofs for p-symmetry, inner p-symmetry, inner p-reflexivity, p-reflexivity, and p-determinism are analogous to the original proofs in the proof of Theorem 28 in \citeA{Antic22}.

    \item p-Commutativity fails since $a-b\neq b-a$ whenever $a\neq b$.

    \item Central permutation follows from the fact that $a-b=c-d$ iff $a-c=b-d$.

    \item Strong inner p-reflexivity follows from the fact that $a-a=c-d$ implies $d=c$

    \item Strong p-reflexivity follows from the fact that $a-b=a-d$ implies $d=b$.

    \item p-Determinism follows from the fact that $a-a=a-d$ iff $d=a$.

    \item p-Transitivity follows from
    \begin{align*} 
        a-b=c-d \quad\text{and}\quad c-d=e-f \quad\Rightarrow\quad a-b=e-f.
    \end{align*}

    \item Inner p-transitivity follows from
    \begin{prooftree}
        \AxiomC{$a-b=c-d$}
        \AxiomC{$b-e=d-f$}
        \BinaryInfC{$a-b+b-e=c-d+d-f$}
        \UnaryInfC{$a-e=c-f$.}
    \end{prooftree}

    \item Central p-transitivity is a direct consequence of p-transitivity. Explicitly, we have
    \begin{align*} 
        a-b=b-c \quad\text{and}\quad b-c=c-d \quad\Rightarrow\quad a-b=c-d.
    \end{align*}
\end{itemize}
\end{proof}

\begin{fact}\label{f:con_+}
    \AxiomC{$a:b::_m c:d$}
        \AxiomC{$e:f::_m g:h$}
        \RightLabel{.}
    \BinaryInfC{$a+e:b+f::_m c+g:d+h$}
    \DisplayProof
\end{fact}
\begin{proof} \hfill
\begin{prooftree}
    \AxiomC{$a:b::_m c:d$}
    \RightLabel{\ref{t:DPT}}
    \UnaryInfC{$a-b=c-d$}
        \AxiomC{$e:f::_m g:h$}
        \RightLabel{\ref{t:DPT}}
        \UnaryInfC{$e-f=g-h$}
    \BinaryInfC{$(a+e)-(b+f)=a-b+e-f=c-d+g-h=(c+g)-(d+h)$}
    \UnaryInfC{$a+e:b+f::_m c+g:d+h$.}
\end{prooftree}
\end{proof}

\prettyref{f:con_+} shows that we can decompose number proportions; for example,
\begin{align*} 
  4:5::0:1 = (2:3::0:1)+(2:2::0:0).
\end{align*}

The following notion of a number proportion is an instance of the more general definition due to \citeA[Proposition 2]{Stroppa06} given for abelian semigroups, adapted to the additive setting of this section:

\begin{definition} For any integers $a,b,c,d\in\mathbb Z$, define (recall that the underlying algebra is $(\mathbb{Z,+,Z})$)
\begin{align*} 
    a:b::_{SY}c:d \quad:\Leftrightarrow\quad &a=k+o,\quad b=\ell+o,\quad c=k+u,\quad d=\ell+u,\quad\text{for some $k,\ell,o,u\in\mathbb Z$.}
\end{align*}
\end{definition}

We have the following equivalence which, surprisingly, shows that the notion of an additive number proportion in \citeA{Stroppa06} coincides with the rather restricted monolinear fragment of our framework and is therefore too weak to be used as a \textit{general} definition of an additive number proportion (and see \prettyref{e:2436}):

\begin{theorem}\label{t:SY=>m} $a:b::_{SY}c:d \quad\Leftrightarrow\quad a:b::_{( \mathbb Z,+, \mathbb Z),m} c:d$.
\end{theorem}
\begin{proof} A direct consequence of the Difference Proportion \prettyref{t:DPT}.
\end{proof}

\begin{example}\label{e:2436} The natural number proportion
\begin{align*} 
    2:4::3:6
\end{align*} is characteristically justified within our framework by $x\to x+x$, which is non-monolinear since $x$ occurs more than once on the right-hand side; on the other hand, this simple proportion is not captured within \citeS{Stroppa06} framework, that is, 
\begin{align}\label{equ:not_models_SY} 
    2:4\not::_{SY} 3:6.
\end{align}
\end{example}



\subsection{Monolinear multiplicative arithmetical proportions}\label{§:MMAPs}

In this subsection, we work in $(\mathbb Q,\cdot, \mathbb Q)$, where we study multiplicative monolinear number proportions. We begin by noting that the set of justifications of $a\to b$ is given by
\begin{align*} 
    \uparrow^m(a\to b)= \{kx\to \ell x\mid\; &a\to b=ko\to\ell o,\text{ for some }k,\ell,o\in\mathbb Q\}\\
        &\cup\{kx\to b\mid a\to b=ko\to b,\text{ for some }k,o\in\mathbb Q\}\cup\{a\to b\}.
\end{align*} This leads to the following characterization of the monolinear entailment relation with respect to multiplication:

\begin{theorem}[Geometric Proportion Theorem]\label{t:GPT} For any $a,b,c,d\in\mathbb Q$,
\begin{align*} 
    a:b::_{(\mathbb Q,\cdot, \mathbb Q),m} c:d \quad&\Leftrightarrow\quad a=ko,\quad b=\ell o,\quad c=ku,\quad d=\ell u,\quad k,\ell,o,u\in\mathbb Q\\
        &\Leftrightarrow\quad \frac a b=\frac c d\quad\text{\textit{\textbf{(geometric proportion)}}}.
\end{align*} The first equivalence holds in $(\mathbb{N,\cdot,N})$ as well.\footnote{This will be essential in \prettyref{§:Primes} when we study primes.}
\end{theorem}
\begin{proof} We first show the second equivalence
\begin{align*} 
    a=ko,\quad b=\ell o,\quad c=ku,\quad d=\ell u \quad\Leftrightarrow\quad \frac a b=\frac c d,
\end{align*} for some $k,\ell,o,u\in\mathbb Q$. The direction from left to right holds trivially. For the other direction, notice that $\frac ab=\frac cd$ implies
\begin{align*} 
    a=\left(\frac cd\right)b, \quad b=1b, \quad c=\left(\frac cd\right)d, \quad d=1d.
\end{align*}

The rest of the proof is similar to the proof of \prettyref{t:DPT}.
\end{proof}

The Geometric Proportion \prettyref{t:GPT} shows that monolinear multiplicative number proportions can be geometrically interpreted as analogical proportions between rectangles. Moreover, the simple characterization of the monolinear relation in \prettyref{t:GPT} allows us to analyze the proportional properties within the monolinear setting:

\begin{theorem}\label{t:properties_cdot} All the properties listed in \prettyref{§:P} hold in $(\mathbb Q,\cdot, \mathbb Q)$ in the monolinear fragment.
\end{theorem}
\begin{proof} We have the following proofs:
\begin{itemize}
    \item The proofs for p-symmetry, inner p-symmetry, inner p-reflexivity, p-reflexivity, and p-determinism are analogous to the original proofs in the proof of Theorem 28 in \citeA{Antic22}.

    \item p-Commutativity follows from \prettyref{t:GPT} together with
    \begin{align*} 
        ko:\ell o::_m \ell o:ko,\quad\text{for all $k,\ell,o\in\mathbb Q$.}
    \end{align*}

    \item Central permutation follows from \prettyref{t:GPT} together with
    \begin{align*} 
        ko:\ell o::_m ku:\ell u\quad\Leftrightarrow\quad ok:uk::_m o\ell:u\ell.
    \end{align*}

    \item Strong inner p-reflexivity follows from \prettyref{t:GPT} together with 
    \begin{align*} 
        ko:ko::_m ku:d \quad\Leftrightarrow\quad d=ku.
    \end{align*}

    \item Strong p-reflexivity follows from \prettyref{t:GPT} together with
    \begin{align*} 
            ko:\ell o::_m ko:d \quad\Leftrightarrow\quad d=\ell o.
    \end{align*}

    \item p-Determinism follows from $\frac a a=\frac a d$ iff $d=a$.

    \item p-Transitivity follows from \prettyref{t:GPT} together with
        \begin{align*} 
            \frac a b=\frac c d \quad\text{and}\quad \frac c d=\frac e f \quad\Rightarrow\quad \frac a b=\frac e f.
        \end{align*}

    \item Inner p-transitivity follows from the following derivation:
    \begin{prooftree}
        \AxiomC{$a:b::_m c:d$}
        \RightLabel{\prettyref{t:GPT}}
        \UnaryInfC{$\frac ab=\frac cd$}
        \UnaryInfC{$a=\frac{bc}d$}
        \AxiomC{$b:e::_m d:f$}
        \RightLabel{\prettyref{t:GPT}}
        \UnaryInfC{$\frac be=\frac df$}
        \UnaryInfC{$e=\frac{bf}d$} 
        \BinaryInfC{$\frac ae=\frac{\frac{bc}d}{\frac{bf}d}=\frac{bcd}{bfd}=\frac cf$}
        \UnaryInfC{$a:e::_m c:f$.}
    \end{prooftree}

    \item Central p-transitivity is an immediate consequence of transitivity.
\end{itemize}
\end{proof}


\begin{theorem}\label{t:con_.}
    \AxiomC{$a:b::_m c:d$}
        \AxiomC{$a':b'::_m c':d'$}
        \RightLabel{.}
    \BinaryInfC{$aa':bb'::_m cc':dd'$}
    \DisplayProof
\end{theorem}
\begin{proof} \hfill
\begin{prooftree}
    \AxiomC{$a:b::_m c:d$}
    \RightLabel{\ref{t:GPT}}
    \UnaryInfC{$ko:\ell o::_m ku:\ell u$}
        \AxiomC{$a':b'::_m c':d'$}
        \RightLabel{\ref{t:GPT}}
        \UnaryInfC{$k'o':\ell' o'::_m k'u':\ell' u'$}
    \BinaryInfC{$(ko)(k'o'):(\ell o)(\ell' o')::_m (ku)(k'u'):(\ell u)(\ell' u')$}
    \UnaryInfC{$(kk')(oo'):(\ell\ell')(oo')::_m (kk')(uu'):(\ell\ell')(uu')$}
    \UnaryInfC{$aa':bb'::_m cc':dd'$.}
\end{prooftree} 
\end{proof}

\subsubsection{Primes}\label{§:Primes}

We shall now prove some properties of the monolinear entailment relation with respect to primes. In this subsection, the underlying algebra is $(\mathbb{N,\cdot,N})$ where $\mathbb N$ denotes the natural numbers.

\begin{proposition} Let $p,q,p',q'$ be primes. We have
\begin{align*} 
    p:q::_m p':q' \quad\Leftrightarrow\quad (p=q \quad\text{and}\quad p'=q') \quad\text{or}\quad (p=p' \quad\text{and}\quad q=q').
\end{align*}
\end{proposition}
\begin{proof} By \prettyref{t:GPT}, we have
\begin{align*} 
    p:q::_m p' :q' \quad\Leftrightarrow\quad p=ko,\quad q=\ell o,\quad c=ko',\quad d=\ell o',\quad\text{for some $k,\ell,o,o'\in\mathbb N$}.
\end{align*} We distinguish two cases. First, if $k=1$ and $o=p$, then $q=\ell p$ which implies $\ell=1$ and therefore $q'=o'$ and $p'=o'$. Second, if $o=p$ and $o=1$, then $q=\ell$ and $p'=po'$ which implies $o'=1$ and therefore $q'=\ell=q$.
\end{proof}

\begin{proposition} Let $p,q$ be primes, and let $c,d\in\mathbb N$. We have
\begin{align*} 
    p:q::_m c:d \quad\Leftrightarrow\quad (p=q& \quad\text{and}\quad c=d) \quad\text{or}\\ 
        &(p\neq q \quad\text{and}\quad c=pu \quad\text{and}\quad d=qu,\text{ for some }u\in\mathbf).
\end{align*}
\end{proposition}
\begin{proof} By \prettyref{t:GPT}, we have
\begin{align*} 
    p :q::_m c:d \quad\Leftrightarrow\quad p=ko,\quad q=\ell o,\quad c=ku,\quad d=\ell u,\quad\text{for some $k,\ell,o,u\in\mathbb N$}.
\end{align*} We distinguish two cases. First, if $k=1$ and $o=p$, then $q=\ell p$ and thus $\ell=1$ and $q=p$ and $c=d=u$. Second, if $k=p$ and $o=1$, then $q=\ell$, $c=pu$, and $d=qu$, for some $u\in\mathbb N$.
\end{proof}

\subsection{Monolinear word proportions}\label{§:MWP}

In this subsection, we work in the word domain $( A^\ast,\cdot, A^\ast)$, where the set of monolinear justifications of $\mathbf{a\to b}$ is given by
\begin{align*} 
    \uparrow^m(\mathbf{a\to b})=\left\{\mathbf a_1x\mathbf a_3\to\mathbf b_1x\mathbf b_3 \;\middle|\;\mathbf a=\mathbf a_1\mathbf a_2\mathbf a_3\to\mathbf b_1\mathbf b_2\mathbf b_3,\;\mathbf a_1,\mathbf a_2,\mathbf a_3,\mathbf b_1,\mathbf b_3\in A^\ast\right\}\\
    \cup\{\mathbf a_1x\mathbf a_3\to\mathbf b\mid\mathbf a=\mathbf a_1\mathbf a_2\mathbf a_3\to\mathbf b;\mathbf a_1,\mathbf a_2,\mathbf a_3\in A^\ast\}\cup\{\mathbf{a\to b}\}.
\end{align*} This implies
\begin{align*} 
    \uparrow^m(\mathbf{a\to b \righttherefore c\to d})=&\left\{\mathbf a_1x\mathbf a_3\to\mathbf b_1x\mathbf b_3 \;\middle|\; 
        \begin{array}{l}
            \mathbf{a\to b}=\mathbf a_1\mathbf a_2\mathbf a_3\to\mathbf b_1\mathbf a_2\mathbf b_3\\
            \mathbf{c\to d}=\mathbf a_1\mathbf c_2\mathbf a_3\to\mathbf b_1\mathbf c_2\mathbf b_3\\
            \mathbf a_1,\mathbf a_2,\mathbf a_3,\mathbf b_1,\mathbf b_3,\mathbf c_2\in A^\ast
        \end{array}
        \right\}\\
            &\qquad\cup\left\{\mathbf a_1x\mathbf a_3\to\mathbf b \;\middle|\;
            \begin{array}{l}
                \mathbf{b=d}\\
                \mathbf{a\to b}=\mathbf a_1\mathbf a_2\mathbf a_3\to\mathbf b\\
                \mathbf{c\to d}=\mathbf a_1\mathbf c_2\mathbf a_3\to\mathbf b\\
                \mathbf a_1,\mathbf a_2,\mathbf a_3,\mathbf c_2\in A^\ast
            \end{array}
            \right\}\cup\{\mathbf{a\to b}\mid\mathbf{a=c,b=d}\}.
\end{align*}


This leads to the following characterization of the monolinear entailment relation:

\begin{theorem}\label{t:m_words} $\mathbf a: \mathbf b::_m \mathbf c: \mathbf d \quad\Leftrightarrow\quad \mathbf a=\mathbf a_1\mathbf a_2\mathbf a_3,\quad\mathbf b=\mathbf b_1\mathbf a_2\mathbf b_3,\quad\mathbf c=\mathbf a_1\mathbf b_2\mathbf a_3,\quad\mathbf d=\mathbf b_1\mathbf b_2\mathbf b_3$.
\end{theorem}
\begin{proof} $(\Rightarrow)$ By assumption, we have $\mathbf a\to \mathbf b\righttherefore_m\, \mathbf c\to \mathbf d$ which holds iff either\footnote{Recall from \prettyref{d:abcd} that $\emptyset$ here really means the set of all trivial monolinear justifications $\emptyset_{(A^+,\cdot,A^+)}$.}
\begin{align*} 
    \uparrow^m(\mathbf{a\to b})\ \cup \uparrow^m(\mathbf{c\to d})=\emptyset,
\end{align*} or $\uparrow^m(\mathbf{a\to b\righttherefore c\to d})$ is non-empty and subset maximal with respect to $d$. In the first case, notice that neither $\uparrow^m(\mathbf{a\to b})$ nor $\uparrow^m(\mathbf{c\to d})$ can be empty since we always have
\begin{align*} 
    \mathbf{a\to b}\in\ \uparrow^m(\mathbf{a\to b}) \quad\text{and}\quad \mathbf{c\to d}\in\ \uparrow^m(\mathbf{c\to d}).
\end{align*} In the second case, by assumption we must have some monolinear justification $s(x)\to t(x)$ of $\mathbf{ a\to b\righttherefore c\to d}$ in $( A^\ast,\cdot, A^\ast)$. We distinguish the following cases:
\begin{enumerate}
    \item If $s(x)\to t(x)$ equals $\mathbf{a\to b}$ or $\mathbf{c\to d}$, we must have
    \begin{align*} 
        \mathbf{a=c} \quad\text{and}\quad \mathbf{b=d}.
    \end{align*}

    \item Else if $s(x)\to t(x)$ equals $\mathbf a_1x\mathbf a_3\to\mathbf b_1x\mathbf b_3$, we must have
    \begin{align*} 
        \mathbf a=\mathbf a_1\mathbf a_2\mathbf a_3 \quad\text{and}\quad \mathbf b=\mathbf b_1\mathbf a_2\mathbf b_3 \quad\text{and}\quad \mathbf c=\mathbf a_1\mathbf b_2\mathbf a_3 \quad\text{and}\quad \mathbf d=\mathbf b_1\mathbf b_2\mathbf b_3
    \end{align*} for some $\mathbf a_1,\mathbf a_2,\mathbf a_3,\mathbf b_1,\mathbf b_2,\mathbf b_3\in A^\ast$.

    \item Else if $s(x)\to t(x)$ equals $\mathbf a_1x\mathbf a_3\to\mathbf b$, we must have $\mathbf a=\mathbf a_1\mathbf a_2\mathbf a_3$, $\mathbf c=\mathbf a_1\mathbf b_2\mathbf a_3$, and $\mathbf d=\mathbf b$, for some $\mathbf a_1,\mathbf a_2,\mathbf a_3,\mathbf b_2\in A^\ast$. Then, by assumption, we must also have $\mathbf a: \mathbf b::_m \mathbf c: \mathbf b$ and, by inner p-symmetry, $\mathbf b: \mathbf a::_m \mathbf b: \mathbf c$ and therefore $\mathbf b\to \mathbf a \righttherefore_m\, \mathbf b\to \mathbf c$. So, either we have $\uparrow^m(\mathbf{b\to a})\ \cup \uparrow^m(\mathbf{b\to c})=\emptyset$ or $\uparrow^m(\mathbf{ b\to a \righttherefore b\to c})$ is non-empty and subset maximal with respect to $\mathbf c$. Again, the sets $\uparrow^m(\mathbf{b\to a})$ and $\uparrow^m(\mathbf{b\to c})$ cannot be empty as they contain $\mathbf{b\to a}$ and $\mathbf{b\to c}$, respectively. Hence, $\uparrow^m(\mathbf{ b\to a \righttherefore b\to c})$ contains at least one non-trivial monolinear justification $s'(x)\to t'(x)$. We distinguish the following cases:
    \begin{enumerate}
        \item If $s'(x)\to t'(x)$ equals $\mathbf{b\to a}$ or $\mathbf{b\to c}$, we must have $\mathbf{a=c}$.

        \item Else if $s'(x)\to t'(x)$ equals $\mathbf b'_1x\mathbf b'_3\to\mathbf a'_1x\mathbf a'_3$, for some $\mathbf{b'_1,b'_3,a'_1,a'_3}\in A^\ast$, we must have $\mathbf b=\mathbf b'_1\mathbf b'_2\mathbf b'_3=\mathbf b'_1\mathbf c'_2\mathbf b'_3$ and $\mathbf a=\mathbf a'_1\mathbf b'_2\mathbf a'_3$ and $\mathbf c=\mathbf a'_1\mathbf c'_2\mathbf a'_3$, for some $\mathbf b'_2,\mathbf c'_2\in A^\ast$. The identity $\mathbf b'_1\mathbf b'_2\mathbf b'_3=\mathbf b'_1\mathbf c'_2\mathbf b'_3$ implies $\mathbf b'_2=\mathbf c'_2$ and again $\mathbf{a=c}$.

        \item Finally, if $s'(x)\to t'(x)$ equals $\mathbf b'_1x\mathbf b'_3\to\mathbf a$, we must also have $\mathbf{a=c}$.
    \end{enumerate} 
\end{enumerate}

$(\Leftarrow)$ The monolinear justification $\mathbf a_1x\mathbf a_3\to\mathbf b_1x\mathbf b_3$ is a characteristic justification of
\begin{align*} 
    \mathbf a_1\mathbf a_2\mathbf a_3\to\mathbf b_1\mathbf a_2\mathbf b_3 \righttherefore\mathbf a_1\mathbf c_2\mathbf a_3\to\mathbf b_1\mathbf c_2\mathbf b_3 
    \quad\text{and}\quad 
    \mathbf a_1\mathbf c_2\mathbf a_3\to\mathbf b_1\mathbf c_2\mathbf b_3 \righttherefore\mathbf a_1\mathbf a_2\mathbf a_3\to\mathbf b_1\mathbf a_2\mathbf b_3
\end{align*} in $( A^\ast,\cdot, A^\ast)$ by the Uniqueness \prettyref{l:UL} since $\mathbf a_1x\mathbf a_3$ and $\mathbf b_1x\mathbf b_3$ both induce injective word functions. Analogously, $\mathbf b_1x\mathbf b_3\to\mathbf a_1x\mathbf a_3$ is a characteristic justification of 
\begin{align*} 
    \mathbf a_1\mathbf c_2\mathbf a_3\to\mathbf b_1\mathbf c_2\mathbf b_3 \righttherefore \mathbf a_1\mathbf a_2\mathbf a_3\to\mathbf b_1\mathbf a_2\mathbf b_3
    \quad\text{and}\quad 
    \mathbf a_1\mathbf a_2\mathbf a_3\to\mathbf b_1\mathbf a_2\mathbf b_3 \righttherefore \mathbf a_1\mathbf c_2\mathbf a_3\to\mathbf b_1\mathbf c_2\mathbf b_3.
\end{align*} Hence, we have shown the theorem.
\end{proof}

\begin{corollary} $\mathbf a: \mathbf{eaf}::_m \mathbf c: \mathbf{ecf}$.
\end{corollary}

\begin{corollary}\label{c:ab=cd} $\mathbf{ab=cd} \quad\not\Rightarrow\quad \mathbf a: \mathbf b::_m \mathbf c: \mathbf d$.
\end{corollary}
\begin{proof} For example, by \prettyref{t:m_words} we have $a:b\not::_m \varepsilon:ab$.
\end{proof}


\begin{corollary}\label{c:^r} There are words such that $\mathbf a:\mathbf a^r\not::_m \mathbf c:\mathbf c^r$.
\end{corollary}
\begin{proof} For example, by \prettyref{t:m_words} we have $ab:ba\not::_m ba:ab$.
\end{proof}




The simple characterization of the monolinear proportion relation in \prettyref{t:m_words} allows us to analyze the proportional properties within the monolinear word setting:

\begin{theorem} The monolinear word proportion relation satisfies
\begin{itemize}
    \item symmetry,
    \item inner p-symmetry,
    \item p-reflexivity,
    \item p-determinism,
    \item strong inner p-reflexivity,
    \item strong p-reflexivity,
    \item transitivity,
    \item central transitivity,
\end{itemize} and, in general, it dissatisfies
\begin{itemize}
    \item central permutation,
    \item commutativity,
    \item inner transitivity.
\end{itemize}
\end{theorem}
\begin{proof} We have the following proofs:
\begin{itemize}
    \item The proofs for p-symmetry, inner p-symmetry, inner p-reflexivity, p-reflexivity, and p-determinism are analogous to the original proofs in the proof of Theorem 28 in \citeA{Antic22}.

    \item p-Determinism is by \prettyref{t:m_words} equivalent to
    \begin{align*} 
        \mathbf a_1\mathbf a_2\mathbf a_3:\mathbf b_1\mathbf a_2\mathbf b_3::_m \mathbf a_1\mathbf b_2\mathbf a_3:\mathbf b_1\mathbf b_2\mathbf b_3 \quad\Leftrightarrow\quad \mathbf b_1\mathbf b_2\mathbf b_3=\mathbf a_1\mathbf a_2\mathbf a_3
    \end{align*} where
    \begin{align*} 
        \mathbf a=\mathbf a_1\mathbf a_2\mathbf a_3=\mathbf b_1\mathbf a_2\mathbf b_3=\mathbf a_1\mathbf b_2\mathbf a_3 \quad\text{and}\quad \mathbf d=\mathbf b_1\mathbf b_2\mathbf b_3.
    \end{align*} This follows from
    \begin{align*} 
        \mathbf a_1\mathbf a_2\mathbf a_3=\mathbf a_1\mathbf b_2\mathbf a_3 \quad\Leftrightarrow\quad \mathbf b_2=\mathbf a_2 \quad\Leftrightarrow\quad \mathbf b_1\mathbf b_2\mathbf b_3=\mathbf a \quad\Leftrightarrow\quad \mathbf{d=a}.
    \end{align*}

    \item Central permutation fails\footnote{See \prettyref{r:Problem30_m}}, for example, given the alphabet $ A:=\{a_1,a_2,a_3,b_1,b_3,c_2\}$ since as a consequence of \prettyref{t:m_words}, we have
    \begin{align*} 
        a_1a_2a_3:b_1a_2b_3::_m a_1c_2a_3:b_1c_2b_3
    \end{align*} whereas
    \begin{align*} 
        a_1a_2a_3:a_1c_2a_3::_m b_1a_2b_3:b_1c_2b_3.
    \end{align*}

    \item Strong inner p-reflexivity and strong p-reflexivity are immediate consequences of \prettyref{t:m_words}.

    \item p-Commutativity fails, for example, in $ A:=\{a,b\}$ since as a consequence of \prettyref{t:m_words}, we have
    \begin{align*} 
        a:b::_m b:a.
    \end{align*}

    \item p-Transitivity is an immediate consequence of \prettyref{t:m_words}.

    \item Inner p-transitivity fails, for example, in $ A:=\{a_1,a_2,a_3,b_1,b_3,c_2,d_2,e_1,e_3\}$ since as a consequence of \prettyref{t:m_words}, we have
    \begin{align*} 
        a_1a_2a_3:b_1a_2b_3::_m a_1c_2a_3:b_1c_2a_3
    \end{align*} and
    \begin{align*} 
        b_1a_2b_3:e_1a_2e_3::_m b_1d_2b3:e_1d_2e_3
    \end{align*} whereas
    \begin{align*} 
         a_1a_2a_3:e_1a_2e_3::_m a_1c_2a_3:e_1d_2e_3.
    \end{align*}

    \item Finally, central p-transitivity is an immediate consequence of transitivity already shown above.
    \end{itemize}
\end{proof}

\begin{remark}\label{r:Problem30_m} The fact that central permutation fails gives a negative answer to Problem 30 in \citeA{Antic22} in the monolinear setting.
\end{remark}


\begin{remark} Notice that we cannot prove an analogue of Theorems \ref{f:con_+} and \ref{t:con_.} in the word domain since by \prettyref{t:m_words}, we in general have
\begin{align*} 
    \mathbf a: \mathbf b::_m \mathbf c: \mathbf d \quad\text{and}\quad \mathbf a': \mathbf b'::_m \mathbf c': \mathbf d' \quad\not\Rightarrow\quad \mathbf{aa'}: \mathbf{bb'}::_m \mathbf{cc'}: \mathbf{dd'}.
\end{align*}
\end{remark}

\section{Word proportions}\label{§:WP}

In \prettyref{§:MWP}, we studied word proportions in the monolinear fragment where justifications contain at most one occurrence of a single variable. This section studies word proportions in the full framework. Specifically, we 
show in \prettyref{t:Stroppa06} and \prettyref{e:Stroppa06} that our framework strictly subsumes \citeS{Stroppa06} notion of a word proportion.

In this section, we shall work in the word algebra $(A^+,\cdot,A^+)$, where $\cdot$ denotes concatenation (always omitted) and $A$ is a finite non-empty alphabet. We will often omit the reference to the word algebra for readability.

\begin{notation} In the other sections, we have used boldface letters to denote sequences of elements. In this section, we we shall use boldface letters to denote words (i.e. sequences of letters) and we use the vector symbol to denote sequences of words.
\end{notation}

\begin{theorem} We have 
\begin{align*} 
     \mathbf a_1\textbf{o}_1\mathbf a_2\ldots&\mathbf a_n\textbf{o}_n\mathbf a_{n+1}:\mathbf b_1\textbf{o}_1\mathbf b_2\ldots\mathbf b_n\textbf{o}_n\mathbf b_{n+1}::\mathbf a_1\mathbf u_1\mathbf a_2\ldots\mathbf a_n\mathbf u_n\mathbf a_{n+1} :\mathbf b_1\mathbf u_1\mathbf b_2\ldots\mathbf b_n\mathbf u_n\mathbf b_{n+1},
\end{align*} for all $\mathbf a_1,\mathbf a_{n+1},\mathbf b_1,\mathbf b_{n+1}\in A^\ast$, $\mathbf a_2,\ldots,\mathbf a_n,\mathbf b_2,\ldots,\mathbf b_n\in A^+$, and $\textbf{o}_i,\mathbf u_i\in A^\ast$, $n\geq 1$.
\end{theorem}
\begin{proof} Since $\mathbf a_1x_1\mathbf a_2\ldots\mathbf a_nx_n\mathbf a_{n+1}$ and $\mathbf b_1x_1\mathbf b_2\ldots\mathbf b_nx_n\mathbf b_{n+1}$ are injective in $A^+$ whenever $\mathbf a_2,\ldots,\mathbf a_n$ and $\mathbf b_2,\ldots,\mathbf b_n$ are non-empty words (notice that two variables side by side yield non-injective word terms), the justification $\mathbf a_1x_1\mathbf a_2\ldots\mathbf a_nx_n\mathbf a_{n+1}\to \mathbf b_1x_1\mathbf b_2\ldots\mathbf b_nx_n\mathbf b_{n+1}$ is a characteristic justification by the Uniqueness \prettyref{l:UL}.
\end{proof}

\begin{corollary}\label{c:a_eaf__c_ecf} $\mathbf a :\mathbf{eaf} :: \mathbf c :\mathbf{ecf} \quad\text{and}\quad \mathbf a :\mathbf b :: \mathbf{eaf} :\mathbf{ebf}$.
\end{corollary}

\begin{proposition}\label{p:r} Given non-empty words $\mathbf{a,b}\in A^+$ of equal length, we have
\begin{align}\label{equ:a_a^r_b_b^r} 
    \mathbf a:\mathbf a^r::\mathbf b:\mathbf b^r.
\end{align} Consequently,
\begin{align*} 
    \mathbf a:\mathbf a^r::\mathbf a^r:\mathbf a.
\end{align*}
\end{proposition}
\begin{proof} Since $|\mathbf a|=|\mathbf b|$ holds by assumption (here $|\,.\,|$ denotes the length of $\mathbf a$), the rule $$x_1\ldots x_{|\mathbf a|}\to x_{|\mathbf a|}\ldots x_1$$ is a characteristic justification of $\mathbf a:\mathbf a^r::\mathbf b:\mathbf b^r$ in $A^+$ by the Uniqueness \prettyref{l:UL}. The second assertion is a consequence of the first together with the identity $\mathbf a^{rr}=\mathbf a$.
\end{proof}

In the proof of \prettyref{p:r}, two assumptions are essential: first, we exclude the empty word since otherwise the Uniqueness \prettyref{l:UL} is no longer applicable; and second, we assume that $\mathbf a$ and $\mathbf b$ have the same length for the same reason. We can avoid both assumptions by simply adding the \textit{injective} reverse operation to our algebra since by the Functional Proportion \prettyref{t:FPT}, we then have in $(A^\ast,\cdot,.^r,A^\ast)$ (notice that we now have included the empty word) the proportions
\begin{align*} 
    \mathbf a:\mathbf a^r::\mathbf b:\mathbf b^r,\quad\text{for \textit{all $\mathbf{a,b}\in A^\ast$}}.
\end{align*} This is a common pattern and we encourage users of our framework to add (preferable injective) algebraic operations to the word algebra as needed for intended applications.

The following notion of word proportion is an instance of \citeA[Definition 2]{Stroppa06}, originally given in the more general context of semigroups: 

\begin{definition}[\citeA{Stroppa06}] Given $\mathbf{a,b,c,d}\in A^+$, define
\begin{align*} 
    \mathbf a: \mathbf b::_{SY}\, \mathbf c: \mathbf d \quad:\Leftrightarrow\quad \mathbf a=a_1\ldots a_n,\quad \mathbf b=b_1\ldots b_n,\quad\mathbf c=c_1\ldots c_n,\quad\mathbf d=d_1\ldots d_n,
\end{align*} for some symbols $a_1,\ldots,a_n,b_1,\ldots,b_n,c_1,\ldots,c_n,d_1,\ldots,d_n\in A$ and $n\geq 0$ such that
\begin{align*}
    (a_i=b_i \quad\text{and}\quad c_i=d_i) \quad\text{or}\quad (a_i=c_i \quad\text{and}\quad b_i=d_i)\quad\text{holds for all $1\leq i\leq n$.}
\end{align*}
\end{definition}

\begin{example} $a:a::_{SY}bb:bb$ and $abc:abd::_{SY}bbc:bbd$.
\end{example}

We have the following implication:

\begin{theorem}\label{t:Stroppa06} $\mathbf a: \mathbf b::_{SY}\, \mathbf c: \mathbf d \quad\Rightarrow\quad \mathbf a: \mathbf b:: \mathbf c: \mathbf d$.
\end{theorem}
\begin{proof} Let $\mathbf a=a_1\ldots a_n$, $\mathbf b=b_1\ldots b_n$, $\mathbf c=c_1\ldots c_n$, and $\mathbf d=d_1\ldots d_n$, $n\geq 1$, be decompositions of $\mathbf{a,b,c,d}$ into letters of the alphabet. Let $I$ be the set of indices $i$ such that $a_i\neq c_i$ and $b_i\neq d_i$
, $1\leq i\leq n$, and let $m$ be the finite cardinality of $I=\{i_1,\ldots,i_m\}$. 

If $m=0$ then we have $\mathbf a=\mathbf c$ and $\mathbf b=\mathbf d$ which together with p-reflexivity of word proportions implies $\mathbf{a:b::c:d}$. 

Otherwise, we define the terms $s$ and $t$ over $A^\ast$\todo{$A^+$?} and $\{x_{i_1},\ldots,x_{i_m}\}$ as follows. For every $i\in [1,n]$, if $a_i=c_i$ and $b_i=d_i$ define $s_i:=a_i$ and $t_i:=c_i$ and, otherwise, define $s_i:=x_i$ and $t_i:=x_i$; finally, define
\begin{align*} 
   s(x_{i_1},\ldots,x_{i_m}):=s_1\ldots s_n \quad\text{and}\quad t(x_{i_1},\ldots,x_{i_m}):=t_1\ldots t_n.
\end{align*} By construction, we have
\begin{align*} 
   s(x_{i_1},\ldots,x_{i_m})&=a_1\ldots a_{i_1-1}x_{i_1}a_{i_1+1}\ldots a_{i_m-1}x_{i_m}a_{i_m+1}\ldots a_n,\\
    t(x_{i_1},\ldots,x_{i_m})&=c_1\ldots c_{i_1-1}x_{i_1}c_{i_1+1}\ldots c_{i_m-1}x_{i_m}c_{i_m+1}\ldots c_n,
\end{align*} and
\begin{align}\label{equ: 2022-12-08-a->b_c->d} 
    \mathbf{a\to b}=s(a_{i_1},\ldots,a_{i_m})\to t(a_{i_1},\ldots,a_{i_m}) \quad\text{and}\quad \mathbf{c\to d}=s(c_{i_1},\ldots,c_{i_m})\to t(c_{i_1},\ldots,c_{i_m}),
\end{align} which shows that $s\to t$ is a justification of $\mathbf{a\to b \righttherefore c\to d}$ in $A^\ast$\todo{$A^+$?}. Notice that if $s$ does 

It remains to show that $s\to t$ is a characteristic justification. For this, we seek to apply the Uniqueness \prettyref{l:UL}. Let $\vec{\mathbf e}=(\mathbf e_1,\ldots,\mathbf e_m)\in (A^\ast)^m$ be an arbitrary sequence of words satisfying $\mathbf a=s(\vec{\mathbf e})$. We need to show $\mathbf b=t(\vec{\mathbf e})$. By \prettyref{equ: 2022-12-08-a->b_c->d}, we have
\begin{align*} 
    \mathbf b=c_1\ldots c_{i_1-1}a_{i_1}c_{i_1+1}\ldots c_{i_m-1}a_{i_m}c_{i_m+1}\ldots c_n.
\end{align*} From $\mathbf a=s(\vec{\mathbf e})$, we deduce
\begin{align*} 
    \mathbf a=a_1\ldots a_{i_1-1}\mathbf e_1a_{i_1+1}\ldots a_{i_m-1}\mathbf e_ma_{i_m+1}\ldots a_n,
\end{align*} which entails
\begin{align*} 
    a_1\ldots a_{i_1-1}a_{i_1}a_{i_1+1}\ldots a_{i_m-1}a_{i_m}a_{i_m+1}\ldots a_n=a_1\ldots a_{i_1-1}\mathbf e_1a_{i_1+1}\ldots a_{i_m-1}\mathbf e_ma_{i_m+1}\ldots a_n.
\end{align*} This implies
\begin{align*} 
    a_{i_1}\ldots a_{i_m}=\mathbf e_1\ldots\mathbf e_m.
\end{align*} From this it is easy to see that
\begin{align*} 
    t(\vec{\mathbf e})=c_1\ldots c_{i_1-1}\mathbf e_1c_{i_1+1}\ldots c_{i_m-1}\mathbf e_mc_{i_m+1}\ldots c_n=\mathbf b.
\end{align*} 

The other cases of Item 4 are analogous.
\end{proof}

Notice that \citeA{Stroppa06} define word proportions only for words over the same alphabet, which means that we  cannot expect the converse of \prettyref{t:Stroppa06} to be true --- the following example shows that it may fail even in the case of a single alphabet:

\begin{example}\label{e:Stroppa06} Let $A:=\{a,b,c\}$. Since $xby$ and $xcy$ are injective in $A^+$\todo{later $A^\ast$ used}, the rewrite rule $$xby\to xcy$$ is a characteristic justification of
\begin{align}\label{equ:ab_ac_bc_cc} 
    ab:ac::bc:cc
\end{align} in $(A^+,\cdot,A^+)$ by the Functional Proportion \prettyref{t:FPT}. This solution formalizes the intuitive observation that $ac$ is obtained from $ab$ by replacing $b$ by $c$ --- analogously, $cc$ is obtained from $bc$ by again replacing $b$ by $c$. This solution cannot be obtained with respect to $::_{SY}$ as the first letter of $ab$ and $ac$ is identical, whereas the first letter of $bc$ and $cc$ differs. (However, we should mention that \prettyref{equ:ab_ac_bc_cc} does hold with respect to $::_{SY}$ in case the empty word is allowed).
\end{example}

\section{Tree proportions}\label{§:TP}

In this section, we study analogical proportions between trees called tree proportions within the free term algebra $\mathfrak T_{L,X}$ (see \prettyref{§:P}).

The following simple observation highlights a peculiar property of the tree setting:

\begin{fact}\label{f:char} Every term function is injective in the term algebra and thus every justification is a characteristic one.
\end{fact} 
\begin{proof} A direct consequence of the Uniqueness \prettyref{l:UL}.
\end{proof}

\begin{example} The following example shows that there may be terms $p,q,r,u,u'$, $u\neq u'$, such that
\begin{align*} 
    p\to q \righttherefore r\to u \quad\text{and}\quad p\to q \righttherefore r\to u'
\end{align*} both hold in $\mathfrak T_{L,X}$:\footnote{Here we have joined two diagrams into one for brevity separated by semicolons.}
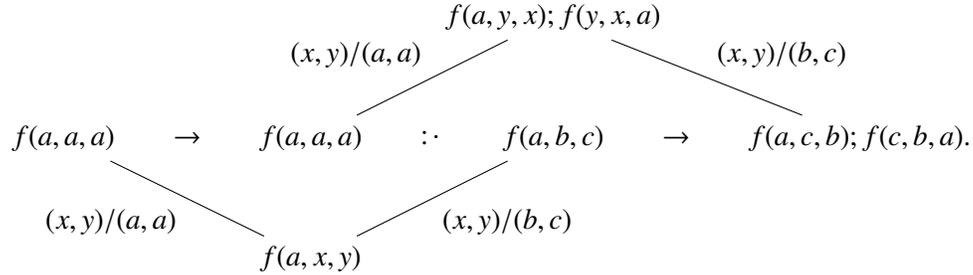
\begin{center}
\begin{tikzpicture}[node distance=1cm and 0.5cm]
\node (a)               {$f(a,a,a)$};
\node (d1) [right=of a] {$\to$};
\node (b) [right=of d1] {$f(a,a,a)$};
\node (d2) [right=of b] {$\righttherefore $};
\node (c) [right=of d2] {$f(a,b,c)$};
\node (d3) [right=of c] {$\to$};
\node (d) [right=of d3] {$f(a,c,b);f(c,b,a)$.};
\node (s) [below=of b] {$f(a,x,y)$};
\node (t) [above=of c] {$f(a,y,x);f(y,x,a)$};

\draw (a) to [edge label'={$(x,y)/(a,a)$}] (s); 
\draw (c) to [edge label={$(x,y)/(b,c)$}] (s);
\draw (b) to [edge label={$(x,y)/(a,a)$}] (t);
\draw (d) to [edge label'={$(x,y)/(b,c)$}] (t);
\end{tikzpicture}
\end{center} The diagram is a compact representation of the fact that
\begin{align*} 
    f(a,x,y)\to f(a,y,x)
\end{align*} is a justification of 
\begin{align*} 
    f(a,a,a)\to f(a,a,a) \righttherefore f(a,b,c)\to f(a,c,b)
\end{align*} and
\begin{align*} 
    f(a,x,y)\to f(y,x,a)
\end{align*} is a justification of
\begin{align*} 
    f(a,a,a)\to f(a,a,a) \righttherefore f(a,c,b)\to f(c,b,a),
\end{align*} and \prettyref{f:char} shows that they are in fact characteristic ones.
\end{example}

Given some $L$-term $p$ and $s\in\ \uparrow p$, by \prettyref{f:char} there is always a \textit{unique} $\textbf{o}\in T_{L,X}^{r(s)}$ such that $p=s(\textbf{o})$ which we will denote by $o(s,p)$, that is,
\begin{align*} 
   s(o(s,p))=p.
\end{align*} Moreover, define, for some $\textbf{o}\in T_{L,X}^k$, $k\geq 1$,
\begin{align*} 
    \uparrow^\textbf{o} q:=\{t\in T_{L,X}\mid q=t(\textbf{o})\}.
\end{align*}

We are now ready to prove a simple characterization of solutions to proportional term equations:

\begin{theorem}\label{t:Sol} $\mathscr S(p\to q \righttherefore r\to \mathfrak x)=\bigcup_{s\in\uparrow(p\Uparrow_\chi r)}(\uparrow^{o(s,p)} q)(o(s,r))$.
\end{theorem}
\begin{proof} Every justification $s\to t$ of $p\to q \righttherefore r\to u$ has the following form:
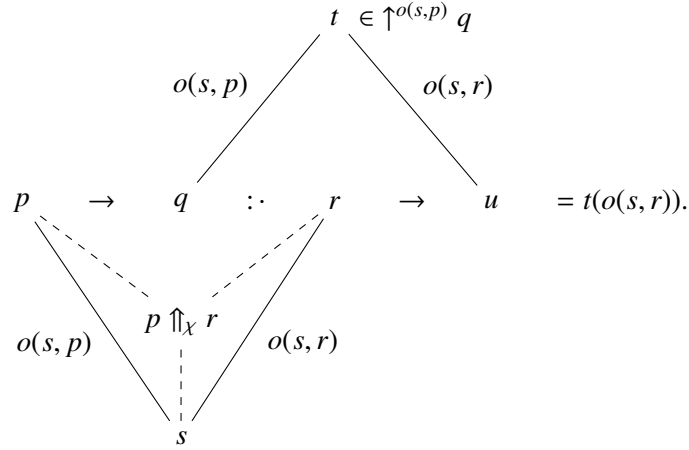
\begin{center}
\begin{tikzpicture}[node distance=1cm and 0.5cm]
\node (a)               {$p$};
\node (d1) [right=of a] {$\to$};
\node (b) [right=of d1] {$q$};
\node (d2) [right=of b] {$\righttherefore $};
\node (c) [right=of d2] {$r$};
\node (d3) [right=of c] {$\to$};
\node (d) [right=of d3] {$u$};
\node (d4) [right=of d] {$=t(o(s,r))$.};
\node (s) [below=of b] {$p\Uparrow_\chi r$};
\node (s') [below=of s] {$s$};
\node (t) [above=of c,yshift=1cm] {$t$};
\node (t') [right=of t,xshift=-0.5cm] {$\in\ \uparrow^{o(s,p)}q$};

\draw[dashed] (a) to  (s); 
\draw (a) to [edge label'={$o(s,p)$}] (s'); 
\draw[dashed] (c) to (s);
\draw[dashed] (s') to (s);
\draw (c) to [edge label={$o(s,r)$}] (s');
\draw (b) to [edge label={$o(s,p)$}] (t);
\draw (d) to [edge label'={$o(s,r)$}] (t);
\end{tikzpicture}
\end{center} Since every justification is a characteristic one by \prettyref{f:char}, every solution $u$ to $p\to q \righttherefore r\to \mathfrak x$ has the form $u=t(o(s,r))$, for some $s\in\ \uparrow(p\Uparrow_\chi r)$ and $t\in\ \uparrow^{o(s,p)}q$.
\end{proof}

For a set of terms $S$ and a term $s$, define
\begin{align*} 
    X_S(s):=X(s)- X(S).
\end{align*} Moreover, we write $s\langle p_i/q_i\mid i\in I\rangle$ for the term which we obtain from $s$ by replacing one or more occurrences of the subterm $p_i$ in $s$ by $q_i$, for every $i\in I$ (this is different from substitutions which replace \textit{all} occurrences of \textit{variables} at once). Notice that we have $s\langle\;\rangle=s$.

\begin{lemma}\label{l:u=r} $u=r\langle p/q\mid p\chi q\in X_{\{p,q,r,u\}}(r\Uparrow_\chi u)\rangle$.
\end{lemma}
\begin{proof} By structural induction on the shape of $u$ and $r$:
\begin{itemize}
    \item $u=a\in L_0$:
        \begin{itemize}
            \item $r=a$: $X_{\{p,q,r,u\}}(r\Uparrow_\chi u)=\emptyset$ and $a=a\langle\;\rangle$.
            \item $r\neq a$: $X_{\{p,q,r,u\}}(r\Uparrow_\chi u)=\{r\chi a\}$ and $a=r\langle r/a\rangle$.
        \end{itemize}
    \item $u=x\in X$:
        \begin{itemize}
            \item $r=x$: $X_{\{p,q,r,u\}}(r\Uparrow_\chi u)=X_{\{p,q,r,u\}}(x\Uparrow_\chi x)=X_{\{p,q,r,u\}}(x)=\emptyset$ and $x=x\langle\;\rangle$.
            \item $r\neq x$: $X_{\{p,q,r,u\}}(r\Uparrow_\chi u)=\{r\chi x\}$ and $x=r\langle r/x\rangle$.
        \end{itemize}
    \item $u=f(u_1,\ldots,u_{r(f)})\in T_{L,X}$:
        \begin{itemize}
            \item $r=f(r_1,\ldots,r_{r(f)})$: By induction hypothesis, we have
            \begin{align*} 
                u_i=r_i\langle p/q\mid p\chi q\in X_{\{p,q,r,u\}}(u_i\Uparrow_\chi r_i)\rangle.
            \end{align*} This implies\todo{check}
            \begin{align*} 
                f&(u_1,\ldots,u_{r(f)})\\
                    &=f(r_1\langle p/q\mid p\chi q\in X(u_1\Uparrow_\chi r_1)\rangle,\ldots,r_{r(f)}\langle p/q\mid p\chi q\in X(u_{r(f)}\Uparrow_\chi r_{r(f)})\rangle)\\
                    &=f(r_1,\ldots,r_{r(f)})\left\langle p/q \;\middle|\;  p\chi q\in\bigcup_{i=1}^{r(f)}X_{\{p,q,r,u\}}(u_i\Uparrow_\chi r_i)\right\rangle.
            \end{align*}

            \item $r=g(r_1,\ldots,r_{rg})$: $X_{\{p,q,r,u\}}(r\Uparrow_\chi u)=\{r\chi u\}$ and $u=r\langle r/u\rangle$.
        \end{itemize}
\end{itemize}
\end{proof}

\begin{lemma}\label{l:ZX_iff_uparrow} $X_{\{p,q,r,u\}}(q\Uparrow_\chi u)\subseteq X_{\{p,q,r,u\}}(p\Uparrow_\chi r) \quad\Leftrightarrow\quad (p\Uparrow_\chi r)\to (q\Uparrow_\chi u)\in\ \uparrow(p\to q \righttherefore r\to u)$.
\end{lemma}
\begin{proof} The direction from right to left holds by definition of justifications. We prove the other direction by nested structural induction on the shape of $p,q,r,u$:
\begin{itemize}
    \item $p=a\in L_0$:
        \begin{itemize}
            \item $r=a$: $X_{\{p,q,r,u\}}(p\Uparrow_\chi r)=X_{\{p,q,r,u\}}(a\Uparrow_\chi a)=\emptyset$ implies $X_{\{p,q,r,u\}}(q\Uparrow_\chi u)=\emptyset$ implies $q=u$. We then have
            \begin{align*} 
                (p\Uparrow_\chi r)\to (q\Uparrow_\chi u)=a\to q\in\ \uparrow(a\to q \righttherefore a\to q).
            \end{align*}

            \item $r\neq a$: $X_{\{p,q,r,u\}}(p\Uparrow_\chi r)=\{ a\chi r\}$:
                \begin{itemize}
                    \item $X_{\{p,q,r,u\}}(q\Uparrow_\chi u)=\emptyset$ implies $q=u$. Hence,
                    \begin{align*} 
                        (p\Uparrow_\chi r)\to (q\Uparrow_\chi u)= a\chi r\to q\in\ \uparrow(a\to q \righttherefore r\to q).
                    \end{align*}

                    \item $X_{\{p,q,r,u\}}(q\Uparrow_\chi u)=\{ a\chi r\}$ implies $u=q\langle s/t\mid s\chi t\in X_{\{p,q,r,u\}}(q\Uparrow_\chi u)\rangle=q\langle a/r\rangle$ (\prettyref{l:u=r}). Hence,\todo{check}
                    \begin{align*} 
                        (p\Uparrow_\chi r)\to (q\Uparrow_\chi u)= a\chi r\to q\langle a/ a\chi r\rangle\in\ \uparrow(a\to q \righttherefore r\to q\langle a/r\rangle).
                    \end{align*}
                \end{itemize}
        \end{itemize}
        \item $p=x\in X$: $X_{\{p,q,r,u\}}(p\Uparrow_\chi r)=\{ x\chi r\}$:
            \begin{itemize}
                \item $X_{\{p,q,r,u\}}(q\Uparrow_\chi u)=\emptyset$ implies $q=u$ and
                \begin{align*} 
                    (p\Uparrow_\chi r)\to (q\Uparrow_\chi u)= x\chi r\to q\in\ \uparrow(x\to q \righttherefore r\to q).
                \end{align*}

                \item $X_{\{p,q,r,u\}}(q\Uparrow_\chi u)=\{ x\chi r\}$ implies $u=q\langle s/t\mid s\chi t\in X_{\{p,q,r,u\}}(q\Uparrow_\chi u)\rangle=q\langle x/r\rangle$ (\prettyref{l:u=r}). Hence,
                \begin{align*} 
                    (p\Uparrow_\chi r)\to (q\Uparrow_\chi u)=x\to q\langle x/ x\chi r\rangle\in\ \uparrow(x\to q \righttherefore r\to q\langle x/r\rangle).
                \end{align*}
            \end{itemize}
        \item $p=f(p_1,\ldots,p_{r(f)})\in T_{L,X}$:
            \begin{itemize}
                \item $r=f(r_1,\ldots,r_{r(f)})$: 

                \begin{itemize}
                    \item $q=a\in L_0$:
                        \begin{itemize}
                            \item $u=a\in L_0$: $(p\Uparrow_\chi r)\to (q\Uparrow_\chi u)=p\Uparrow_\chi r\to a\in\ \uparrow(p\to a \righttherefore r\to a)$.

                            \item $u\neq a$: $X_{\{p,q,r,u\}}(a\Uparrow_\chi u)=\{ a\chi u\}\subseteq X_{\{p,q,r,u\}}(p\Uparrow_\chi r)$ implies $ a\chi u\in X_{\{p,q,r,u\}}(p\Uparrow_\chi r)$. By induction hypothesis, we have $$p_i\Uparrow_\chi r_i\to a\chi u\in\ \uparrow(p_i\to a \righttherefore r_i\to u),$$ for all $1\leq i\leq r(f)$. This holds iff 
                            \begin{align*} 
                                p_i&=(p_i\Uparrow_\chi r_i)(\textbf{o}_i),\\
                                r_i&=(p_i\Uparrow_\chi r_i)(\textbf{o}'_i),
                            \end{align*} for some $\textbf{o}_i=(o_{i,1},\ldots,o_{i,|X_{...}(p_i\Uparrow_\chi r_i)|})\in T_{L,X}^{|X_{...}(p_i\Uparrow_\chi r_i)|}$\todo{``...'' in index} such that $o_{i,j}= a\chi u$, for some $1\leq j\leq |X_{...}(p_i\Uparrow_\chi r_i)|$\todo{``...'' in index}. We then have:
                            \begin{center}
                            \begin{tikzpicture}[node distance=1.5cm and 0cm]
                                \node (a)               {$f(p_1,\ldots,p_{r(f)})$};
                                \node (d1) [right=of a] {$\to$};
                                \node (b) [right=of d1] {$a$};
                                \node (d2) [right=of b] {$\righttherefore $};
                                \node (c) [right=of d2] {$f(r_1,\ldots,r_{r(f)})$};
                                \node (d3) [right=of c] {$\to$};
                                \node (d) [right=of d3] {$u$.};
                                \node (s) [below=of b] {$f(p_1\Uparrow_\chi r_1,\ldots,p_{r(f)}\Uparrow_\chi r_{r(f)})(x\mid x\in X_{\{p,q,r,u\}}(p_i\Uparrow_\chi r_i),1\leq i\leq r(f))$};
                                \node (t) [above=of c] {$ a\chi u$};

                                \draw (a) to [edge label'={$\mathbf x/(\textbf{o}_1,\ldots,\textbf{o}_{r(f)})$}] (s); 
                                \draw (c) to [edge label={$\mathbf x/(\textbf{o}_1',\ldots,\textbf{o}_{r(f)}')$}] (s);
                                \draw (b) to [edge label={$\mathbf x/(\textbf{o}_1,\ldots,\textbf{o}_{r(f)})$}] (t);
                                \draw (d) to [edge label'={$\mathbf x/(\textbf{o}_1',\ldots,\textbf{o}_{r(f)}')$}] (t);
                            \end{tikzpicture}
                            \end{center} This shows
                            \begin{align*} 
                                (p\Uparrow_\chi r)\to (q\Uparrow_\chi u)=f(p_1\Uparrow_\chi r_1,\ldots,p_{r(f)}\Uparrow_\chi r_{r(f)})\to a\chi u\in\ \uparrow(p\to a \righttherefore r\to u).
                            \end{align*}
                        \end{itemize}

                        \item $q=x\in X$:
                            \begin{itemize}
                                \item $u=x$: $(p\Uparrow_\chi r)\to x\in\ \uparrow(p\to z \righttherefore r\to z)$.
                                \item $u\neq x$: $X_{\{p,q,r,u\}}(q\Uparrow_\chi u)=\{ x\chi u\}\subseteq X_{\{p,q,r,u\}}(p\Uparrow_\chi r)$ implies $ x\chi u\in X_{\{p,q,r,u\}}(p\Uparrow_\chi r)$. Now proceed as in the case ``$u\neq a$'' above.
                            \end{itemize}

                        \item $q=g(q_1,\ldots,q_{rg})$:
                            \begin{itemize}
                                \item $u=q$: trivial.
                                \item $u\neq q$: $X_{\{p,q,r,u\}}(q\Uparrow_\chi u)=\{ q\chi u\}\subseteq X_{\{p,q,r,u\}}(p\Uparrow_\chi r)$ implies $ q\chi u\in X_{\{p,q,r,u\}}(p\Uparrow_\chi r)$. Now proceed as in the case ``$u\neq a$'' above.
                            \end{itemize}

                    \item $r=g(r_1,\ldots,r_{rr})$:\footnote{Notice that $r$ means here two things: (i) an $L$-term; and (ii) the rank function. Hence, $rr$ stands for the rank of the term $r$.} $p\Uparrow_\chi r= p\chi r$ and $X_{\{p,q,r,u\}}(q\Uparrow_\chi u)\subseteq\{ p\chi r\}$. We distinguish two cases:
                        \begin{enumerate}
                            \item $X_{\{p,q,r,u\}}=\emptyset$ implies $q=u$.
                            \item $X_{\{p,q,r,u\}}=\{ p\chi r\}$ implies $u=q\langle p/r\rangle$ (\prettyref{l:u=r}). Hence,
                            \begin{align*} 
                                 p\chi r\to q\langle p/ p\chi r\rangle\in\ \uparrow(p\to q \righttherefore r\to u).
                            \end{align*}
                        \end{enumerate}
                \end{itemize}
            \end{itemize}
\end{itemize}
\end{proof}

\begin{lemma}\label{l:p_q_r_u} $X_{\{p,q,r,u\}}(q\Uparrow_\chi u)\subseteq X_{\{p,q,r,u\}}(p\Uparrow_\chi r)$ implies $p\to q \righttherefore_{\mathfrak T_{L,X}}\, r\to u$.\todo{iff?}
\end{lemma}
\begin{proof} A direct consequence of \prettyref{f:char} and \prettyref{l:ZX_iff_uparrow}.
\end{proof}

We have thus arrived at the following sufficient condition for tree proportions:

\begin{theorem}\label{t:pqru} $X_{\{p,q,r,u\}}(q\Uparrow_\chi u)=X_{\{p,q,r,u\}}(p\Uparrow_\chi r)$ implies $p:q::_{\mathfrak T_{L,X}}r:u$.\todo{iff?}
\end{theorem}
\begin{proof} A direct consequence of \prettyref{l:p_q_r_u} and the definition of a tree proportion in terms of arrow proportions.
\end{proof}


\begin{corollary} For every term function $f:T_{L,X}\to T_{L,X}$,
\begin{align*} 
    p:q::_{\mathfrak T_{L,X}} r:u \quad\Leftrightarrow\quad f(p):f(q)::_{\mathfrak T_{L,X}} f(r):f(u).
\end{align*}
\end{corollary}
\begin{proof} A direct consequence of \prettyref{t:pqru} and the fact that
\begin{align*} 
    X_{\{p,q,r,u\}}(q\Uparrow_\chi u) &= X_{\{p,q,r,u\}}(p\Uparrow_\chi r)\\ 
        \quad&\Leftrightarrow\quad X_{\{f(p),f(q),f(r),f(u)\}}(f(q)\Uparrow_\chi f(u))=X_{\{f(p),f(q),f(r),f(u)\}}(f(p)\Uparrow_\chi f(r)).
\end{align*}
\end{proof}

\todo[inline]{Example}

\section{Anti-unification}\label{§:AU}

In the previous section, we have seen that classical syntactic anti-unification can be used to compute tree proportions and motivated by that result, we initiate in this section the study of interactions between analogical proportions and anti-unification beyond the free term algebras. For this, we first need a notion of anti-unification generalized from term to arbitrary algebras which has been recently provided by the author \cite{Antic23-19} and which we shall now briefly recall:

\begin{definition} Define the set of \textit{\textbf{minimally general generalizations}} (or \textit{\textbf{mggs}}) of two elements $a\in A$ and $b\in B$ in $\mathfrak{(A,B)}$ by
\begin{align*} 
    a \Uparrow_{ \mathfrak{(A,B)}} b:=\min_{\sqsubseteq_{ \mathfrak{(A,B)}}}(a\uparrow_{ \mathfrak{(A,B)}} b),
\end{align*} where
\begin{align*} 
    a\uparrow_{ \mathfrak{(A,B)}} b:=(\uparrow_\mathfrak A a)\cap (\uparrow_\mathfrak B b)
\end{align*} and
\begin{align*} 
    \uparrow_ \mathfrak A a:=\left\{s\in T_{L,X} \;\middle|\; a=s^ \mathfrak A(\textbf{o}),\text{ for some $\textbf{o}\in A^{r(s)}$}\right\}.
\end{align*} In case $a\Uparrow_{ \mathfrak{(A,B)}} b=\{s\}$ contains a single generalization, we call $s$ the \textit{\textbf{least general generalization}} of $a$ and $b$ in $\mathfrak{(A,B)}$.
\end{definition}

We now wish to \textit{initiate} the study of connections between algebraic anti-unification and analogical proportions with an illustrative example:

\begin{example}\label{e:20_4_30_x} In \citeA[Example 66]{Antic22}, the author has computed the solutions of
\begin{align*} 
    20:4::30: \mathfrak x
\end{align*} in the multiplicative algebra $ \mathfrak M:=( \mathbb N_2,\cdot, \mathbb N_2)$, where $\mathbb N_2:=\{2,3,\ldots\}$, as
\begin{align*} 
    \mathscr S_ \mathfrak M(20:4::30: \mathfrak x)=\{6,9\}.
\end{align*} The two solutions are characteristically justified respectively by
\begin{align*} 
    10x\to 2x \quad\text{and}\quad 10x\to x^2.
\end{align*} We see that in both cases, the generalization $10x$ occurs on the left-hand side of the respective justifications --- this seems not to be an accident as we have
\begin{align*} 
    20\Uparrow 30=\{10x\}
\end{align*} as we are now going to show; in other words, $10x$ is the least general generalization of $20$ and $30$ in $\mathfrak M$.

Recall from \citeA[Example 66]{Antic22} that we have (we omit the subscript $\mathfrak M$)
\begin{align*}
    \uparrow 4&=\{4,2x,xy,x^2,x\},\\
    \uparrow 20&=\left\{
    \begin{array}{ccc}
        20 & 10x & 5x^2\\
         & 4x & 5xy\\
        2xy & & xyz\\
        2x & 5x & xy\\
        x^2y & x
    \end{array}
    \right\},\\
    \uparrow 30&=\left\{
    \begin{array}{ccc}
        30 & 15x & \\
        10x & 6z & 5xy\\
        2xy & 3xy & xyz\\
        2x & 5x & xy\\
        & 3x & x
    \end{array}
    \right\}.
\end{align*} Hence, we have
\begin{align*} 
    20\uparrow 30=\{10x,2xy,5xy,xyz,2x,5x,xy,x\}.
\end{align*} We now want to find the $\sqsubseteq$-minimal generalizations in $20\Uparrow 30$. The generalization $x$ is, of course, not minimal. We have
\begin{align*} 
    \downarrow 2xy=\{8,12,18,\ldots\}\sqsubset\{4,6,8,10,12,\ldots\}=\ \downarrow 2x
\end{align*} which shows
\begin{align*} 
    2x\sqsubset 2xy.
\end{align*} Hence, we can exclude $2xy$. Analogously,
\begin{align*} 
    5x\sqsubset 5xy.
\end{align*} shows that we can exclude $5xy$. Moreover, we clearly have
\begin{align*} 
    10x\sqsubset 2x \quad\text{und}\quad 10x\sqsubset 5x,
\end{align*} which means that we can exclude $2x$ and $5x$. We are thus left with the generalizations
\begin{align*} 
    10x \quad\text{and}\quad xy \quad\text{and}\quad xyz.
\end{align*} We clearly have
\begin{align*} 
    10x\sqsubset xy \quad\text{und}\quad 10x\sqsubset xyz,
\end{align*} which means that we are left with
\begin{align*} 
    20\Uparrow 30=\min_\sqsubseteq{(20\uparrow 30)}=\{10x\}.
\end{align*}
\end{example}

Of course, a single example is not enough to establish a strong connection between two concepts --- this brief section is to be understood only as an \textit{inspiration} for a deeper investigation of the relation between algebraic anti-unification and analogical proportions.

\section{Finite algebras}\label{§:FA}

In this section, we provide an algorithm for the computation of the analogical proportion relation within the $k$-fragment (cf. \prettyref{§:kl}) in finite algebras via tree automata.

Recall that a (\textit{\textbf{frontier-to-root}}) \textit{\textbf{tree automaton}} \cite<see e.g.>{Gecseg15}
\begin{align*} 
    \mathscr T_{\mathfrak A,k,\alpha,F}:=(\mathfrak A,L,X_k,\alpha,F)
\end{align*} consists of
\begin{itemize}
    \item a \textit{finite} $L$-algebra $\mathfrak A$,
    \item an \textit{\textbf{initial assignment}} $\alpha:X_k\to A$, and
    \item a set $F\subseteq A$ of \textit{\textbf{final states}}.
\end{itemize} The \textit{\textbf{regular tree language}} recognized by $\mathscr T_{\mathfrak A,k,\alpha,F}$ is given by
\begin{align*} 
    ||\mathscr T_{\mathfrak A,k,\alpha,F}||:=\left\{s\in T_{L,X_k} \;\middle|\; s^ \mathfrak A\alpha\in F\right\}.
\end{align*} We can thus rewrite the set of $k$-generalizations of $a$ in $\mathfrak A$ (see \prettyref{§:kl}),
\begin{align*} 
    \uparrow^k_ \mathfrak A a := (\uparrow_ \mathfrak A a)\cap T_{L,X_k},
\end{align*} by
\begin{align*} 
    \uparrow^k_ \mathfrak A a=\bigcup_{\alpha\in A^{X_k}}||\mathscr T_{\mathfrak A,k,\alpha,\{a\}}||,
\end{align*} and the set of $k$-justifications of an arrow $a\to b$ in $\mathfrak A$ in terms of tree automata as
\begin{align*} 
    \uparrow^k_\mathfrak A(a\to b)=\bigcup_{\alpha\in A^{X_k}}\left[(||\mathscr T_{\mathfrak A,k,\alpha,\{a\}}||\to ||\mathscr T_{\mathfrak A,k,\alpha,\{b\}}||)\cap\{s\to t\mid X(t)\subseteq X(s)\}\right],
\end{align*} where for two forests $S$ and $T$,
\begin{align*} 
    S\to T:=\{s\to t\mid s\in S,t\in T\}.
\end{align*} Notice that the set $A^{X_k}$ of all initial assignments $X_k\to A$ is \textit{finite} by our assumption that $X_k$ is a finite set of variables. Since it is well-known that tree automata are closed under finite unions, there is a tree automaton
\begin{align*} 
    \mathscr T_{\mathfrak A,k,a\to b},
\end{align*} for every arrow $a\to b$, such that
\begin{align*} 
    \uparrow^k_\mathfrak A(a\to b)=||\mathscr T_{\mathfrak A,k,a\to b}||\cap \{s\to t\mid X(t)\subseteq X(s)\}.
\end{align*} Now since tree automata are closed under intersection as well, there is a tree automaton
\begin{align*} 
    \mathscr T_{\mathfrak{(A,B)},k,a\to b \righttherefore c\to d},
\end{align*} for every arrow proportion $a\to b \righttherefore c\to d$, such that
\begin{align*} 
    \uparrow^k_\mathfrak{(A,B)}(a\to b \righttherefore c\to d)=||\mathscr T_{\mathfrak{(A,B)},k,a\to b \righttherefore c\to d}||\cap\{s\to t\mid X(t)\subseteq X(s)\}.
\end{align*} Since $\mathfrak B$ is finite by assumption, checking the $d$-maximality of $\uparrow^k_\mathfrak{(A,B)}(a\to b \righttherefore c\to d)$ can thus be easily achieved with a search linear in the size of $\mathfrak B$; checking the emptiness of $\uparrow^k_\mathfrak A(a\to b)$ and $\uparrow^k_\mathfrak B(c\to d)$ is well-known to be decidable as well \shortcite<cf.>[p. 40]{Comon08}; and checking $X(t)\subseteq X(s)$ is a simple syntactic comparison, which in total gives us an algorithm for deciding
\begin{align*} 
    a\to b \righttherefore_{\mathfrak{(A,B)},k}\, c\to d
\end{align*} and thus for deciding
\begin{align*} 
    a:b::_{\mathfrak{(A,B)},k}c:d.
\end{align*}

\begin{pseudocode}[Decision algorithm]\label{pseudocode:decision_algorithm} Given $k\geq 1$, $a,b\in A$, $c,d\in B$, and a pair of finite $L$-algebras $\mathfrak{(A,B)}$, we design an algorithm for deciding whether
\begin{align*} 
    a:b\stackrel{?}{::}_{ \mathfrak{(A,B)},k}c:d.
\end{align*} We first provide an algorithm for deciding whether
\begin{align}\label{equ:a_to_b__c_to_d} 
    a\to b\stackrel?{\righttherefore}_{ \mathfrak{(A,B)},k}\, c\to d.
\end{align}
\begin{enumerate}
    \item Construct the tree automata $\mathscr T_{\mathfrak A,k,a\to b}$ and $\mathscr T_{\mathfrak B,k,c\to d}$ as described above. If 
    \begin{align*} 
        (||\mathscr T_{\mathfrak A,k,a\to b}||\cup ||\mathscr T_{\mathfrak B,k,c\to d}||)\cap\{s\to t\mid X(t)\subseteq X(s)\}=\emptyset
    \end{align*} then stop with answer ``yes''.
    \item Otherwise, construct the tree automaton $\mathscr T_{\mathfrak{(A,B)},k, a\to b\righttherefore c\to d}$ as described above and compute the forest
    \begin{align*} 
        T := ||\mathscr T_{\mathfrak{(A,B)},k, a\to b\righttherefore c\to d}||.
    \end{align*}
    \item Compute the set of $k$-justifications
    \begin{align*} 
        J := T-\{s\to t\mid X(t)\not\subseteq X(s)\}
    \end{align*} applying a simple syntactic check on each rule in $T$ (notice that $T$ may be infinite).
    \item We now want to check whether $J$ is $d$-maximal:
    \begin{enumerate}
	    \item For each $d'\neq d\in B$: 
	        \begin{enumerate}
	            \item Construct the set of $k$-justifications
	            \begin{align*} 
	                J':=||\mathscr T_{\mathfrak{(A,B)},k, a\to b\righttherefore c\to d'}||-\{s\to t\mid X(t)\not\subseteq X(s)\}
	            \end{align*} as above.
	            \item If $J\subsetneq J'$ then stop with answer ``no''.
	        \end{enumerate}
	    \item Return the answer ``yes''.
	\end{enumerate}
	\item At this point, we have decided \prettyref{equ:a_to_b__c_to_d}. Now repeat the above steps for the remaining arrow proportions 
	\begin{align*} 
	    b\to a \righttherefore d\to c \quad\text{and}\quad c\to d \righttherefore a\to b \quad\text{and}\quad d\to c \righttherefore b\to a
	\end{align*} and return ``yes'' iff the answer is ``yes'' in each case.
\end{enumerate} 
\end{pseudocode}

\begin{pseudocode}\label{pseudocode:abcd} Given a pair of $L$-algebras $\mathfrak{(A,B)}$ and $k\geq 1$, computing the set
\begin{align*} 
    \{(a,b,c,d)\in A^2\times B^2\mid a:b::_{\mathfrak{(A,B)},k}c:d\}
\end{align*} can be done using \prettyref{pseudocode:decision_algorithm} to decide $a:b\stackrel?{::}_{ \mathfrak{(A,B)},k}c:d$ on each of the finitely many tuples.
\end{pseudocode}

\begin{pseudocode}[Solving proportional equations] Given a proportional equation
\begin{align*} 
    a:b::_k c: \mathfrak x,
\end{align*} finding some/all $d\in B$ such that $a:b::_{\mathfrak{(A,B)},k}c:d$ can be achieved using \prettyref{pseudocode:abcd}.
\end{pseudocode} 


\section{Conclusion}

The purpose of this paper was to expand the mathematical theory of analogical proportions within the abstract algebraic framework recently introduced by the author within the general setting of universal algebra.

We shall now discuss some lines of potential future research.

Section \prettyref{§:kl} introduced the $(k,\ell)$-fragments and in \prettyref{§:MAAP} -- \prettyref{§:MWP} we have studied the monolinear $(1,1)$-fragment in the setting of numbers and words. The monolinear fragment for sets appears more challenging given that union and intersection are highly non-injective operations. The next step is to study the \textbf{linear fragment} $(\infty,1)$ --- consisting of justifications with arbitrary many variables occurring at most once --- for numbers, words, and sets. From a theoretical point, it is interesting to analyze relationships between different such fragments as it may be the case that we obtain a \textbf{$(k,\ell)$-hierarchy} of increasing expressive power.

In \prettyref{§:WP}, we have derived some first partial results regarding word proportions where we showed in \prettyref{t:Stroppa06} that an important notion of word proportions is subsumed by our framework. However, a \textbf{full characterization of the word proportion relation} --- analogous to the one provided in the monolinear fragment in \prettyref{t:m_words} --- remains a challenging and practically relevant open problem.

\textbf{Infinite trees} naturally arise in the study of programming languages \cite<see e.g.>{Courcelle83}. It is thus interesting to generalize the concepts and results of \prettyref{§:TP} to tree proportions between infinite trees.

The algorithms of \prettyref{§:FA} regarding analogical proportions in finite algebras are restricted to the $k$-fragment where justifications may contain at most $k$ different variables and is thus bounded. This is necessary for the techniques of the theory of tree automata to be applicable. It appears challenging to \textbf{generalize} those \textbf{algorithms for finite algebras to the full framework} where the number of variables in justifications is unbounded. Even more challenging is the task to \textbf{find algorithms for the computation of analogical proportions beyond finite algebras}, most importantly in finitely representable algebras such as automatic structures \cite{Blumensath00,Blumensath04}.

\bibliographystyle{theapa}
\bibliography{/Users/christianantic/Bibdesk/Bibliography,/Users/christianantic/Bibdesk/Publications_J,/Users/christianantic/Bibdesk/Publications_C,/Users/christianantic/Bibdesk/Preprints,/Users/christianantic/Bibdesk/Submitted,/Users/christianantic/Bibdesk/Notes}
\if\isdraft1
\newpage

\section{Proportional homomorphisms}

\section{Analogies}

\fi
\end{document}